\documentclass[11pt]{article}
\usepackage[utf8]{inputenc}
\pdfoutput=1
\usepackage{setspace}
\usepackage{palatino}
\usepackage{graphicx}
\usepackage{float}
\usepackage{titling} 
\usepackage{multirow}
\usepackage{lscape}
\usepackage{amsmath}
\usepackage{amssymb}
\usepackage[a4paper, total={6in, 9.5in}]{geometry}
\usepackage{array}
\usepackage[normalem]{ulem}
\fontfamily{ppl}\selectfont 
\usepackage{eqparbox}
\usepackage{arydshln}
\usepackage{mathrsfs}
\usepackage[american]{babel}
\usepackage{mathtools}

\usepackage{rotating}
\usepackage{amsthm}
\newtheorem{theorem}{Theorem}
\newtheorem{lemma}{Lemma}

\newtheorem{definition}{Definition}

\usepackage[ruled,vlined]{algorithm2e}

\usepackage{bm}
\usepackage{caption}
\usepackage{subcaption}

\usepackage[natbibapa]{apacite}
\PassOptionsToPackage{hyperindex,breaklinks}{hyperref}
\usepackage{hyperref}
\usepackage{xcolor}
\definecolor{darkblue}{rgb}{0, 0, 0.5}
\hypersetup{colorlinks=true,citecolor=darkblue, linkcolor=darkblue, urlcolor=darkblue}

\usepackage[textsize=small]{todonotes}

\title{Unmixing highly mixed grain size distribution data via maximum volume constrained end member analysis
}

\author{Qianqian Qi\\
   {\small\raggedright Hangzhou Dianzi University, China}\\
   \href{mailto:q.qi@hdu.edu.cn}{\texttt{q.qi@hdu.edu.cn}} 
\and Zhongming Chen\\
    {\small\raggedright Hangzhou Dianzi University, China}\\
\href{mailto:zmchen@hdu.edu.cn}{\texttt{zmchen@hdu.edu.cn}}
\and Peter G. M. van der Heijden\\
    {\small\raggedright Utrecht University, the Netherlands and University of Southampton, UK}\\
\href{mailto:p.g.m.vanderheijden@uu.nl}{\texttt{p.g.m.vanderheijden@uu.nl}}
    }

\predate{}
\postdate{}
\date{}
\date{\vspace{-5ex}}

\begin{document}
{\setstretch{.8}
\maketitle
\begin{abstract}

End member analysis
(EMA) unmixes grain size distribution (GSD) data into a mixture of end members (EMs), thus helping understand sediment provenance and depositional regimes and processes. In highly mixed data sets, however, many EMA algorithms
find EMs which are still a mixture of true EMs.  To overcome this, we propose maximum volume constrained EMA (MVC-EMA), which finds EMs as different as possible. We provide a uniqueness theorem and a quadratic programming algorithm for MVC-EMA. Experimental results show that MVC-EMA can effectively find true EMs in highly mixed data sets.


\end{abstract}
}



\section{Introduction}\label{S: int}

In sedimentary geology, unmixing grain size distribution (GSD) data into a mixture of end members (EMs), called end member analysis (EMA), can help understand sediment provenance and depositional
regimes and processes  \citep{renner1995construction, weltje1997end, paterson2015new, van2018genetically, Liu2023Universal, moskalewicz2024identification, Lin2025Using, RENNY2026113384}. Denote $\bm{P}$ as a data matrix of size $I\times J$. Each row of $\bm{P}$ represents an observed specimen of GSD data with nonnegative and sum-to-1 constraints: $\bm{P} \geq 0$ and $\bm{P}\bm{1} = \bm{1}$, where $\geq$ means that each elements in the matrix is nonnegative and $\bm{1}$ is the vector of all ones of appropriate dimension. EMA approximates a GSD data $\bm{P}$ by a lower rank matrix $\bm{\Pi}$ which is the product of two nonnegative matrices with row-sum-to-1 constraint. Thus, EMA can be expressed as
\begin{equation}\label{E: emamodnoiseless}
\begin{split}
     \text{min     } ~~~&\frac{1}{2}||\bm{P} - \bm{\Pi}||_F^2\\
    \text{subject to    }  ~~~& \bm{\Pi} = \bm{W}\bm{G},~~~\bm{W}\bm{1} = \bm{1},~~~  \bm{G}\bm{1} = \bm{1}\\
    & \bm{W}\in \Re^{I\times K}_+,~~~ \bm{G}\in \Re^{K\times J}_+
\end{split}    
\end{equation}
\noindent where $K \leq \text{min}\{I, J\}$ and the Frobenius norm of a matrix $\bm{A}$ is defined by $||\bm{A}||_F^2 = \sqrt{\sum_{i}\sum_{j}\bm{A}(i, j)^2}$. In addition to Frobenius norm, there are other objective function such as $L_1$ norm to measure the discrepancy between $\bm{P}$ and $\bm{\Pi}$ \citep{ZHANG2020106656}.
Each row of $\bm{G}$ is a end-member, and each specimen in $\bm{\Pi}$ is constructed by these end-members with weights or abundances in the corresponding row of $\bm{W}$. $\bm{W}$ and $\bm{G}$ are called the abundance matrix and the end-member matrix, respectively. Geometrically, the rows of $\bm{\Pi}$ are  in the convex hull generated by the rows of $\bm{G}$ \citep{avis1995good, gillis2020nonnegative}. Given an observed matrix $\bm{P}$, EMA aims to recover $\bm{W}$ and $\bm{G}$ that yield the lower rank matrix $\bm{\Pi}$.

EMA, however, is not unique \citep{weltje1997end, weltje2007genetically, renner1993resolution, renner1995construction, paterson2015new, ZHANG2020106656}. Given EMA $\bm{\Pi} = \bm{W}\bm{G}$, we have
\begin{equation*}
    \bm{\Pi} = \bm{W}\bm{U}\bm{U}^{-1}\bm{G},
\end{equation*}
where $\bm{U}$ is a full rank matrix of size $K\times K$. The sum-to-one conditions $\bm{W}\bm{U}\bm{1} = \bm{1}$ and $\bm{U}^{-1}\bm{G}\bm{1} = \bm{1}$ hold if and only if $\bm{U1} = \bm{1}$ \citep{de1990latent}. The nonnegativity constraints $\bm{W}\bm{U} \ge 0$ and $\bm{U}^{-1}\bm{G} \ge 0$ further restrict the admissible set of $\bm{U}$. However, even under these constraints, $\bm{U}$ may be more than permutation, implying that EMA admits multiple equivalent solutions of the form $(\bm{W}\bm{U}, \bm{U}^{-1}\bm{G})$. Geometrically, nonuniquenes can be  understood as that any convex hull enclosing all data points could be solutions of EMA.

Many EMA algorithms choose to identify the solution by using minimum EMs, i.e. EMs that enclose specimens as tightly as possible \citep{weltje1997end, weltje2007genetically, paterson2015new, van2018genetically, Dietze2022application}. However, this framework fails for highly mixed data where no single specimen is near a true end-member. This is because estimated EMs are themselves mixtures of the true EMs. Highly mixed data, however, are common in natural settings. Following \citet{paterson2015new}, we illustrate a highly mixed case using a simulated GSD data ($99 \times 100$). The data are derived from two lognormal EMs, with abundances always not less than $0.13$ (Figure~\ref{F: True End Members two sources},~\ref{F: True Abundances two sources}). Applying end member modelling algorithm (EMMA) \citep{weltje1997end, seidel2015r}, a fundamental minimum-EMs algorithm, produces two bimodal EMs (Figure~\ref{F: EMMAtwosourcemaxendmem}). Each EM is dominated by a source but is contaminated by the other, leading to abundance misestimation (Figure~\ref{F: EMMAtwosourcemaxabundances}).

\begin{figure}[h]
\centering
\begin{subfigure}[b]{0.45\linewidth}
\includegraphics[width=1\textwidth]{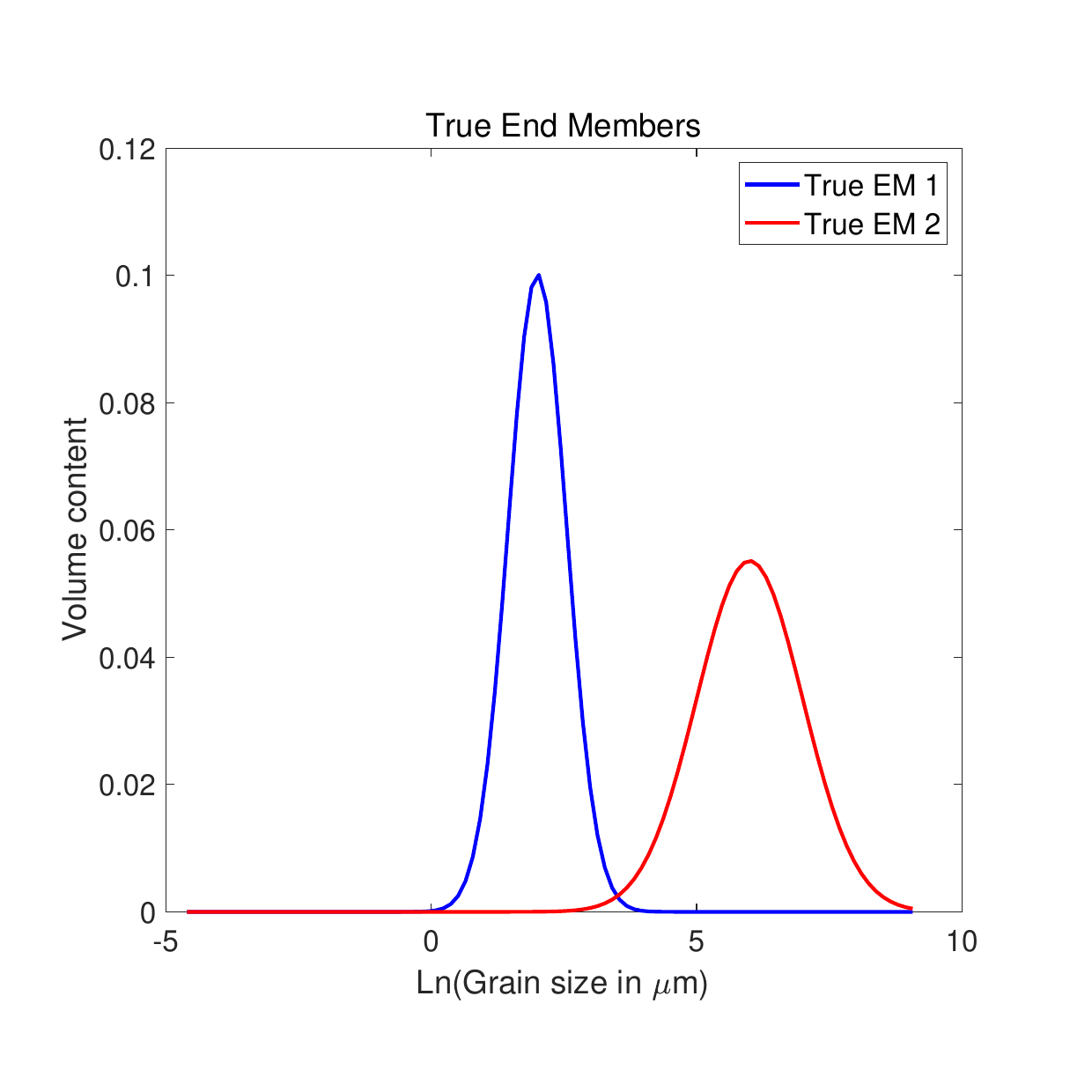}
 \caption{True EMs}\label{F: True End Members two sources}
 \end{subfigure}
 \begin{subfigure}[b]{0.45\linewidth}
\includegraphics[width=1\textwidth]{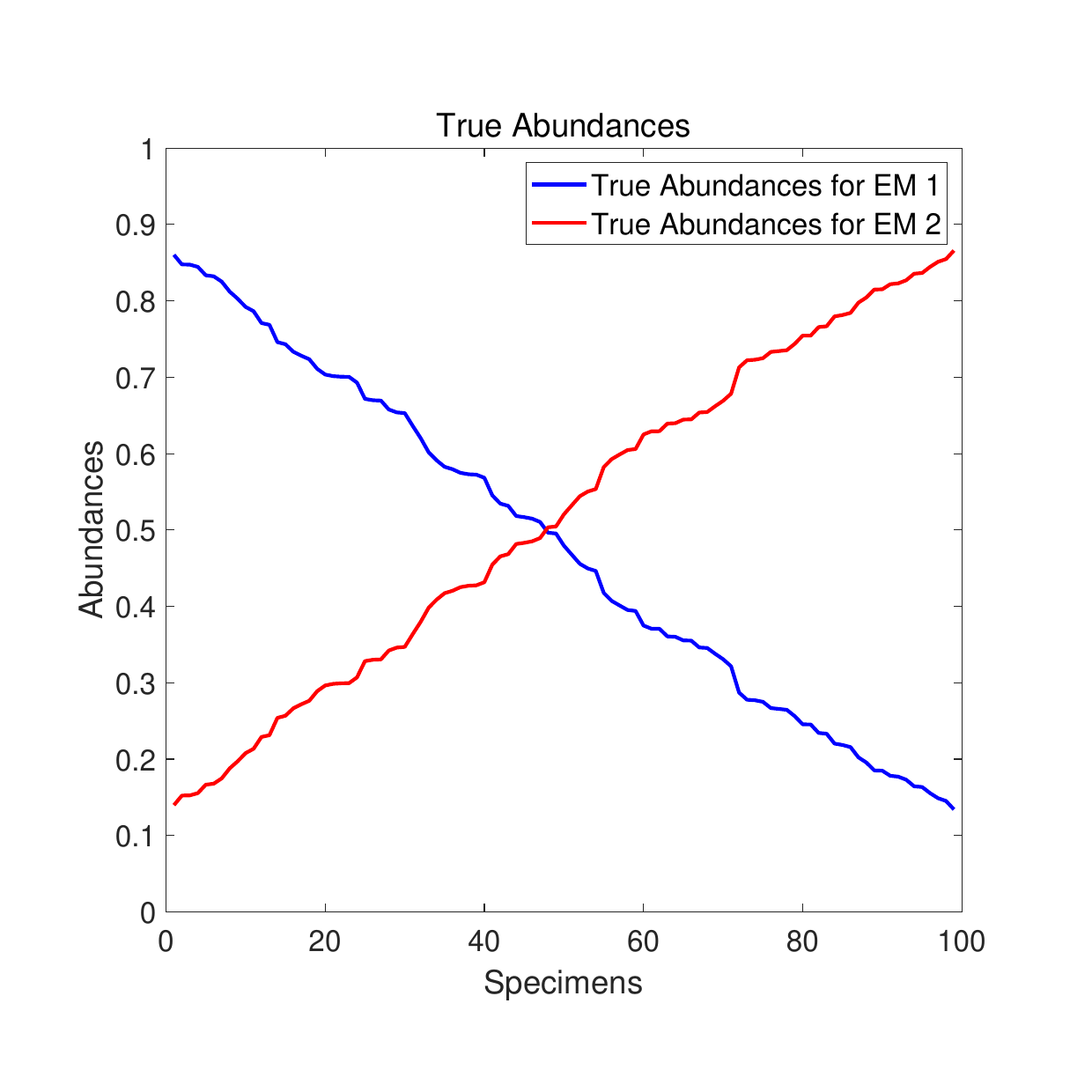}
 \caption{True Abundances}\label{F: True Abundances two sources}
 \end{subfigure}
 \\
   \begin{subfigure}[b]{0.45\linewidth}
\includegraphics[width=1\textwidth]{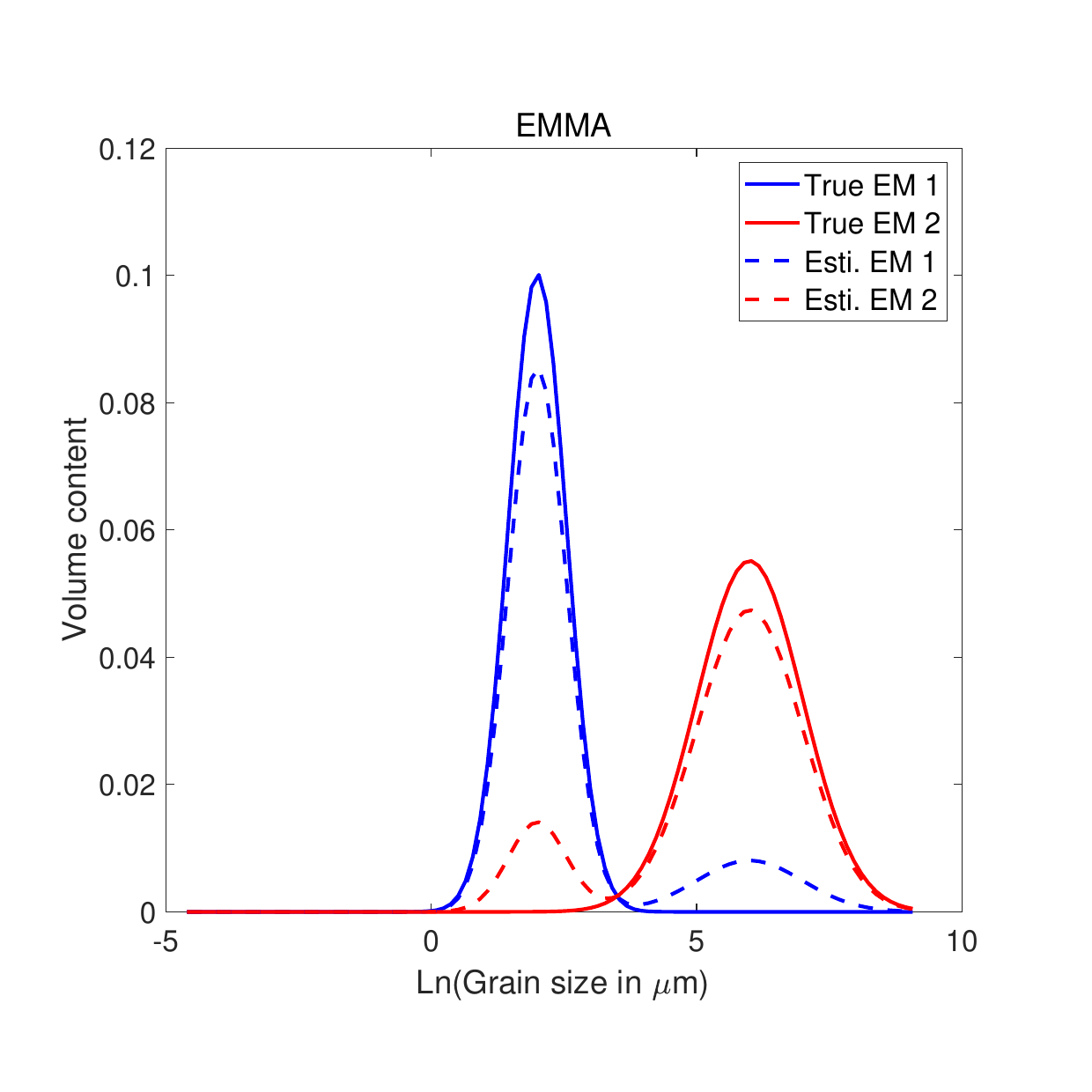}
    \caption{EMMA EMs}
    \label{F: EMMAtwosourcemaxendmem}
 \end{subfigure}
    \begin{subfigure}[b]{0.45\linewidth}
\includegraphics[width=1\textwidth]{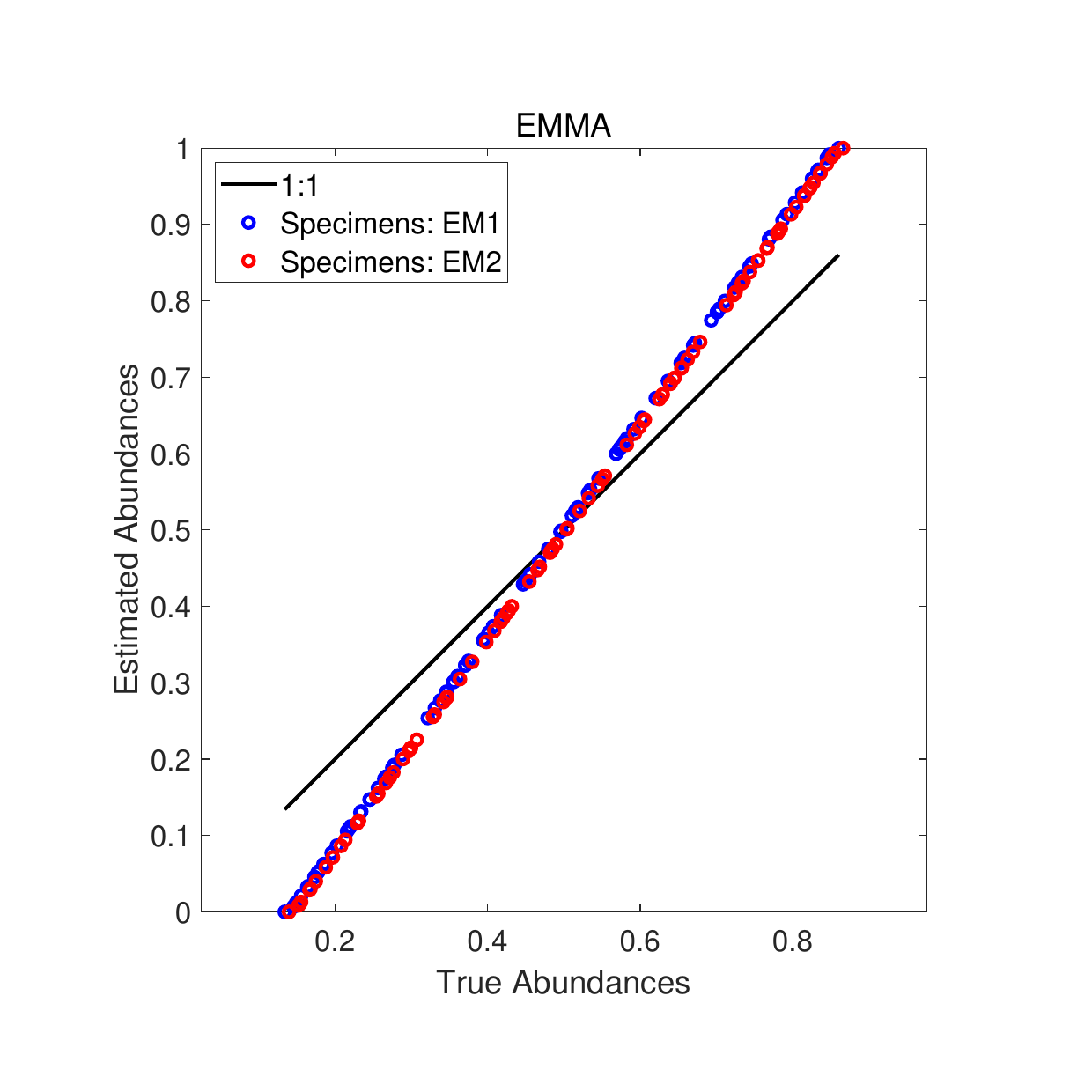}
    \caption{EMMA Abundances}
    \label{F: EMMAtwosourcemaxabundances}
 \end{subfigure}
 \caption{(a and b) Two lognormal EMs and their abundances; (c and d) EMMA unmixing results.}\label{F: emmabasemmatwosource}
    \end{figure}

The "outermost EMs" concept was proposed for highly mixed data by \citet{ZHANG2020106656}, analogous to the "outer extreme solution" in latent budget analysis (LBA) \citep{van1999identifiability}. Notably, LBA has the same parametric form $\bm{\Pi} = \bm{W}\bm{G}$ as EMA \citep{clogg1981latent, de1990latent, heijden1994end}. LBA was developed in the social sciences which can explain pattern of time allocation. However, the framework "outermost EMs" or "outer extreme solution" lacks uniqueness theorems. The only known uniqueness condition, from \citet{de1990latent}, is limited to dimensionality $K=2$, while $K>2$ regularly happens in practice. On the algorithmic side, to find outermost EMs, \citet{ZHANG2020106656} make use of a genetic algorithm. The genetic algorithm is an flexible and widely used algorithm, inspired by natural selection mechanisms. However, due to its stochastic and heuristic nature, it has no guarantee to obtain the global optima and even local optima. \citet{van1999identifiability} proposed Metropolis algorithm to seek the outer extreme solution \citep{RJ-2018-026}. Likewise, the Metropolis algorithm is heuristic in nature.

The "outermost EMs" or "outer extreme solution" is defined in terms of the maximum of the sum of distances between EMs such as the sum of Manhattan or of chi-squared distances \citep{van1999identifiability, ZHANG2020106656}. In contrast, the minimum volume assumption - closely related to "minimum EMs" - has been studied intensively in nonnegative matrix factorization (NMF), that decomposes a nonnegative matrix into the product of two nonnegative matrices, and that includes EMA as a special case \citep{Paatero1994Positive, lee1999learning, Fu2019Nonnegative, gillis2020nonnegative, hobolth2020unifying, GUO2024102379, SaberiMovahed2025Nonnegative, QP2025}.
NMF using the minimum volume assumption has been studied in considerable detail for uniqueness theorems and algorithms, which, partly, benefits from the fact that the volume measure is related to the determinant of $\bm{GG}^T$ rather than related to distances \citep{lin2015identifiability, fu2015, fu2018identifiability, leplat2020, hobolth2020unifying, Abdolali2024Dual, GUO2024102379, SaberiMovahed2025Nonnegative}. 

According to \citet{fu2016robust, Fu2019Nonnegative}, the sum of squared distances between all the pairs of basis vectors $\bm{G}$ is an approximation of the volume of basis matrix $\bm{G}$. In the simulation part of the paper by \citet{fu2016robust}, the regularizier related to the sum of squared distances between basis vectors is slightly worse than the regularizer related to the volume of basis matrix $\bm{G}$ in terms of the mean-squared-error (MSE) of abundance. However, to the best of knowledge, a volume regularizer has never been
used in EMA before.

Motivated by these insights, we propose maximum volume constrained end member analysis (MVC-EMA) for highly mixed data. Section~\ref{S: uniquenesstheorem} proposes a sufficient condition for MVC-EMA to be unqiue. Section~\ref{S: algorithms} proposes an algorithm, named APFGM short for alternative projected fast gradient methods. This algorithm can perform minimum, no, and maximum volume assumptions. Section~\ref{S: experiments} compares APFGM with maximum volume, APFGM with no volume,  APFGM with minimum volume, and  EMMA (a classic algorithm for "Minimum EMs" in sedimentary geology). Finally, Section~\ref{S: conclusion} concludes this paper.

\section{Uniqueness theorem under the maximum volume assumption}\label{S: uniquenesstheorem}

Before beginning this section, we introduce two lemmas in linear algebra, which are used later in this section \citep{gillis2020nonnegative}. The cone of a matrix $\bm{A}\in \Re^{I\times J}$ is defined by
\begin{equation*}
    \text{cone}(\bm{A}) = \left\{\sum_j r_j\bm{A}(:, j) \ |\  r_j \geq 0, \  j = 1, \cdots, J\right\}.
\end{equation*}

\begin{lemma}
    Given a matrix $\bm{A}$, the dual of $\text{cone}(\bm{A})$ is defined by $\text{cone}^*(\bm{A}) = \{\bm{y} | \bm{A}^T\bm{y} \geq 0\}$. 
\end{lemma}

\begin{lemma}
    Given two matrices $\bm{A}$ and $\bm{B}$, if $\text{cone}(\bm{A}) \subseteq \text{cone}(\bm{B})$, then $\text{cone}^*(\bm{B}) \subseteq \text{cone}^*(\bm{A})$.
\end{lemma}

\begin{definition}
    The EMA solution ($\bm{W}$, $\bm{G}$) of $\bm{\Pi}$ is said to be essentially unique if and only if any other EMA solution ($\tilde{\bm{W}}$, $\tilde{\bm{G}}$) has the form
    \begin{equation*}
        \tilde{\bm{W}} = \bm{W}\bm{\Gamma}^{-1} \text{ and } \tilde{\bm{G}} = \bm{\Gamma}\bm{G}
    \end{equation*}
    where $\bm{\Gamma}$ is a permutation matrix.
\end{definition}

Given the lower rank matrix $\bm{\Pi} \in \Re_{+}^{I \times J}$ with $\bm{\Pi}\bm{1} = \bm{1}$, one wants to find an essentially unique solution of EMA $\bm{\Pi} = \bm{WG}$ under the assumption of maximum volume. Formally, the maximum volume constrained end member analysis (MVC-EMA) is formulated as follows:
    \begin{equation}\label{E: maxvolcriteria}
\begin{split}
    \text{max}\quad & \text{det}(\bm{G}\bm{G}^T)\\
    \text{subject to} \quad  & \bm{\Pi} = \bm{W}\bm{G},~~~ \bm{W}\bm{1} = \bm{1},~~~  \bm{G}\bm{1} = \bm{1}\\
    & \bm{W}\in \Re^{I\times K}_+, ~~~ \bm{G}\in \Re^{K\times J}_+  
\end{split}    
\end{equation}
where $K \leq \text{min}\{I, J\}$ and $\text{det}(\bm{G}\bm{G}^T)$ refers to the determinant of $\bm{G}\bm{G}^T$. 

A sufficiently scattered condition (SSC) is a sufficient condition for minimum volume constrained nonnegative matrix factorization (NMF) to be unique \citep{fu2015, fu2018identifiability, leplat2020}. In minimum volume constrained NMF, SSC is used in the context of the coefficient matrix/abundance matrix. Here, we introduce SSC to MVC-EMA. In contrast to minimum volume constrained NMF, SSC is used in the context of the basis matrix/end-member matrix in MVC-EMA. SSC is related to sparsity. This implies that in MVC-EMA the end-member matrix $\bm{G}$, instead of abundance matrix $\bm{W}$, tends to be sparse. 

To derive the uniqueness of the solution of MVC-EMA (\ref{E: maxvolcriteria}), we need the following assumptions as in \citet{fu2015, fu2018identifiability, leplat2020}. 

{\bf Assumption A1}: The matrices $\bm{W}$ and $\bm{G}$ satisfy $\text{rank}(\bm{W}) = \text{rank}(\bm{G}) = K$. 

{\bf Assumption A2}: $\bm{G} \in \Re_+^{K \times J}$ is satisfied with sufficiently scattered conditions (SSC): 

(1) SSC1: $\mathbb{C} \subseteq \text{cone}(\bm{G})$, where $\mathbb{C} = \{\bm{x} \in \Re_{+}^K | \bm{1}^T\bm{x} \geq \sqrt{K - 1}||\bm{x}||_2\}$ is second-order cone;

(2) SSC2: $\text{cone}^*(\bm{G}) \cap bd\mathbb{C}^* = \{\alpha\bm{e}_k | \alpha \geq 0, k = 1, \cdots, K\}$, where the dual of $\mathbb{C}$ is defined as $\mathbb{C}^* = \{\bm{x} \in \Re^K | \bm{1}^T\bm{x} \geq ||\bm{x}||_2\}$.

We show that Assumptions A1 and A2 are a sufficient condition for the uniqueness of the solution of MVC-EMA.
\begin{theorem}\label{theorem: maxvol}
Under Assumptions A1 and A2, MVC-EMA uniquely identifies $\tilde{\bm{W}}$ and $\tilde{\bm{G}}$ up to permutation, i.e., any optimal solution $\bm{W}_{\#}$ and $\bm{G}_{\#}$ to MVC-EMA (\ref{E: maxvolcriteria}) takes the form:
\begin{equation}
    \bm{W}_{\#} = \tilde{\bm{W}}\bm{U}^{-1} \text{ and } \bm{G}_{\#} = \bm{U}\tilde{\bm{G}}
\end{equation}
\noindent where $\bm{U}$ is a permutation matrix.
\end{theorem}

\begin{proof}
Step 1: Let us consider both ($\tilde{\bm{W}}, \tilde{\bm{G}}$) and ($\bm{W}_{\#}$, $\bm{G}_{\#}$) to be optimal solutions for (\ref{E: maxvolcriteria}). Since $\text{rank}(\tilde{\bm{W}}) = \text{rank}(\tilde{\bm{G}}) = \text{rank}(\bm{W}_{\#}) = \text{rank}(\bm{G}_{\#}) = K$, there exists a full rank matrix $\bm{U}$ of size $K\times K$ such that 

\begin{equation*}
    \bm{W}_{\#} = \tilde{\bm{W}}\bm{U}^{-1} \text{ and } \bm{G}_{\#} = \bm{U}\tilde{\bm{G}}.
\end{equation*}
The matrix $\bm{U}$ has row-sum-to-1 constraint because
\begin{equation}\label{eq: sumto1}
    \bm{1} = \bm{G}_{\#}\bm{1} = \bm{U}\tilde{\bm{G}}\bm{1} = \bm{U}\bm{1}
\end{equation}

Step 2: From $\bm{G}_{\#} = \bm{U}\tilde{\bm{G}}$, the vectors in rows of $\bm{U}$ belong to dual cone of $\tilde{\bm{G}}$, i.e., $\bm{U}(i,:)^T \in \text{cone}^*(\tilde{\bm{G}})$ for $i = 1, \cdots, I$ where $\bm{U}(i,:)$ is a row vector with elements being the $i$th row of $\bm{U}$. According to SSC1 ($\mathbb{C} \subseteq \text{cone}(\tilde{\bm{G}})$), we have $\text{cone}^*(\tilde{\bm{G}}) \subseteq \mathbb{C}^*$. Thus, $\bm{U}(i, :)^T \in \mathbb{C}^*$. This means
\begin{equation}\label{eq: dualc}
    ||\bm{U}(i, :)^T||_2 \leq \bm{1}^T\bm{U}(i, :)^T
\end{equation}
Therefore,
\begin{equation*}
    |\text{det}(\bm{U})| = |\text{det}(\bm{U}^T)| \leq \prod_i||\bm{U}(i, :)^T||_2 \leq \prod_i \bm{1}^T\bm{U}(i,:)^T = 1.
\end{equation*}
where the first inequality is the Hadamard's inequality, the second follows (\ref{eq: dualc}), the last equality follows (\ref{eq: sumto1}).

Step 3:  If $|\text{det}(\bm{U})| < 1$, then 
\begin{equation*}
\text{det}(\bm{G}_{\#}\bm{G}_{\#}^T) = \text{det}(\bm{U}\tilde{\bm{G}}\tilde{\bm{G}}^T\bm{U}^T) = |\text{det}(\bm{U})|^2\text{det}(\tilde{\bm{G}}\tilde{\bm{G}}^T) < \text{det}(\tilde{\bm{G}}\tilde{\bm{G}}^T).
\end{equation*}
This means that $\bm{G}_{\#}$ does not have the maximum volume, which is contradict with that $(\bm{W}_{\#}, \bm{G}_{\#})$ is an optimal solution.

Step 4: If $|\text{det}(\bm{U})| = 1$, then all inequalities needs to be equalities. Hence, for all $i$,   $||\bm{U}(i, :)^T||_2 = \bm{1}^T\bm{U}(i, :)^T$, implying that $\bm{U}(i, :)^T$ is the boundary of $\mathbb{C}^*$. And $\bm{U}(i, :)^T \in \text{cone}^*(\tilde{\bm{G}})$. According to SSC2, we have $\bm{U}(i, :)^T = \alpha \bm{e}_k$.

Step 5: And due to the constraint $\bm{U}\bm{1} = \bm{1}$ in Equation~(\ref{eq: sumto1}), we have $\alpha = 1$. Thus $\bm{U}$ can only be a permutation matrix.
\end{proof}

It is worth noting that the maximum volume assumption and the proof process of Theorem~\ref{theorem: maxvol} are similar to the minimum volume assumption and the uniqueness theorem in minimum volume constrained NMF \citep{fu2015, fu2018identifiability, leplat2020}, but the maximum volume assumption and Theorem~\ref{theorem: maxvol} are novel, provide a new perspective, specifically for highly mixed data. Note that the SSC constraint holds for the basis matrix $\bm{G}$ instead of coefficient matrix $\bm{W}$.

We provide an illustration for SSC using  Figure~\ref{F: halfpurity} for $K = 3$
\citep{gillis2020nonnegative, QP2025}. Figure~\ref{F: halfpurity} assumes that viewer stands in the nonnegative orthant, faces the origin, and looks at the two-dimensional plane $\bm{x1} = \bm{1}$ \citep{gillis2020nonnegative}. Specifically, the dots "o" correspond to columns of $\bm{G}$; the crosses "X" are standard basis vectors $\bm{e}_1, \bm{e}_2$, $\bm{e}_3$, the circle corresponds to the second-order cone $\mathbb{C}$; the triangle corresponds to nonnegative orthant $\text{cone}(\bm{e}_1, \bm{e}_2, \bm{e}_3)$; the polygon is $\text{cone}(\bm{G})$. The circle is contained in the polygon. I.e., SSC1 holds.

Geometrically, SSC2 means that the only orthogonal matrix $\bm{Q}$ such that $\text{cone}(\bm{G}) \subseteq \text{cone}(\bm{Q})$ is a permutation matrix \citep{gillis2020nonnegative}. The cone
generated by the columns of any permutation matrix is exactly the triangle in the figure. Any orthogonal matrix is a rotated version of the triangle in the figure. As shown in the figure, the rotated version of the nonnegative orthant $\text{cone}(\bm{e}_1, \bm{e}_2, \bm{e}_3)$ (except for itself) does not contain the polygon formed by the
dots. Thus, SSC2 hold.

\begin{figure}[h]
\centering
\includegraphics[width=0.55\textwidth]{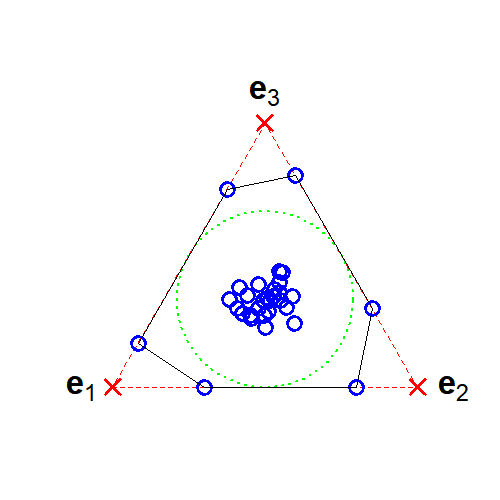}
\caption{Geometric illustrations that end-member matrix $\bm{G}$ satisfies SSC.}\label{F: halfpurity}
\end{figure}

In Figure~\ref{F: halfpurity}, some column points in $\bm{G}$ are located on the edge of the triangle, implying there is a zero element in these columns. Thus, the end-member matrix $\bm{G}$ is sparse. In contrast, in minimum volume constrained NMF, the abundance matrix $\bm{W}$ is satisfied with SSC, and thus the abundance matrix $\bm{W}$ is sparse.

Next, we provide a quadratic programming algorithm for MVC-EMA, inspired by the NMF literature \citep{Zhou2012Minimum, Leplat2019rank, GillisSNPA2014, gillis2020nonnegative}.

\section{APFGM algorithm}\label{S: algorithms}

In this paper, we consider MVC-EMA in which the observed data matrix $\bm{P}$ is approximated by $\bm{\Pi} = \bm{WG}$. The approximation error, quantified by the squared Frobenius norm $||\bm{P} - \bm{WG}||_F^2$, serves as the data fitting term, while $\text{det}(\bm{GG}^T)$ acts as a volume regularizer. The resulting objective function for MVC-EMA is given by
\begin{equation}\label{E: emamodnoise}
\begin{split}
     \text{min}_{\bm{W}, \bm{G}} ~~~& J(\bm{W}, \bm{G})  = \frac{1}{2}||\bm{P} - \bm{WG}||_F^2 - \frac{\lambda}{2} \text{det}(\bm{GG}^T)\\
     \text{subject to } ~~~&\bm{W}\bm{1} = \bm{1},~~~  \bm{G}\bm{1} = \bm{1}, ~~~\bm{W}\in \Re^{I\times K}_+, ~~~ \bm{G}\in \Re^{K\times J}_+.
\end{split}
\end{equation}
\noindent The regularization parameter $\lambda \geq 0$ controls the tradeoff between the data fitting term $||\bm{P} - \bm{WG}||_F^2$ and the volume regularizer $\text{det}(\bm{GG}^T)$. Because of the negative sign before $\lambda$, increasing $\lambda$ results in an increase in the volume. This negative sign is a crucial difference between MVC-EMA and minimum volume constrained NMF, where the sign before $\lambda$ is positive \citep{Zhou2012Minimum}. 

For (\ref{E: emamodnoise}), simultaneously optimizing $\bm{W}$ and $\bm{G}$ is a non-convex problem. As in most work in NMF, we minimize the objective function alternatively over $\bm{W}$ or $\bm{G}$, each time optimizing over one matrix while keeping the
other one fixed. The iteration scheme of alternative optimization can be written as:
\begin{subequations}\label{Eq: l1normalter}
    \begin{align}
       &\bm{W} = \text{arg min}_{\bm{W}: \bm{W}\bm{1} = \bm{1}, \bm{W} \geq 0} J(\bm{W}, \bm{G}) \label{Eq: wl1normalter}\\
        &\bm{G} = \text{arg min}_{\bm{G}: \bm{G}\bm{1} = \bm{1}, \bm{G} \geq 0} J(\bm{W}, \bm{G})
        \label{Eq: gl1normalter}
    \end{align}
\end{subequations}

\subsection{Updating the abundance matrix W}

In Equation~(\ref{Eq: wl1normalter}) $\bm{G}$ is considered known. Given the known $\bm{G}$, the problem of Equation~(\ref{Eq: wl1normalter}) becomes $\bm{W} = \text{arg min}_{\bm{W}: \bm{W}\bm{1} = \bm{1}, \bm{W} \geq 0} \frac{1}{2}||\bm{P} - \bm{WG}||_F^2$ and is therefore a convex optimization problem \citep{gillis2020nonnegative}. Solving $\text{min}_{\bm{W}: \bm{W}\bm{1} = \bm{1}, \bm{W} \geq 0} \frac{1}{2}||\bm{P} - \bm{WG}||_F^2$ is equivalent to solving subproblem
\begin{equation}
    \begin{split}
      \text{min}_{\bm{W}(i, :)}~~~& \frac{1}{2}||\bm{P}(i, :) - \bm{W}(i, :)\bm{G}||_F^2
        \\
        \text{subject to } ~~~& \bm{W}(i, :)\bm{1} = \bm{1}, \bm{W}(i, :) \geq 0
    \end{split}
\end{equation}
\noindent for $i = 1, \cdots, I$, where $\bm{P}(i, :)$ or $\bm{W}(i, :)$ is a row vector with elements being the $i$th row of $\bm{P}$ or $\bm{W}$. For each $i$, it is equivalent to solving
\begin{equation}\label{eq: eachw}
    \begin{split}
      \text{min}_{\bm{W}(i, :)}~~~& \bm{W}(i, :)\left(\frac{1}{2}\bm{GG}^T\right)\bm{W}(i, :)^T - \bm{P}(i, :)\bm{G}^T\bm{W}(i, :)^T
        \\
        \text{subject to } ~~~& \bm{W}(i, :)\bm{1} = \bm{1}, \bm{W}(i, :) \geq 0
    \end{split}
\end{equation}
which is a quadratic programming problem with nonnegative and row-sum-to-one constraints. In practice, the $I$ subproblems can be solved in parallel. If the rank of $\bm{G}$ is $K$, $\frac{1}{2}\bm{GG}^T$ is positive definite. Thus, problem (\ref{eq: eachw}) can be solved by quadratic programming algorithm and the solution is unique \citep{Zhou2012Minimum}. Here, for updating $\bm{W}(i, :)$, following \citet{GillisSNPA2014, Leplat2019rank}, we use a projected fast gradient method (PFGM) \citep{nesterov2004introductory}. 

\subsection{Updating the end-member matrix G}

In Equation~(\ref{Eq: gl1normalter}) $\bm{W}$ is a known matrix. As in updating abundance matrix, we express the objective function of Equation~(\ref{Eq: gl1normalter}) as $K$ subproblems where each one is a quadratic programming problem. 

Data fitting term $||\bm{P} - \bm{W}\bm{G}||_F^2$ can be expressed as \citep{Zhou2012Minimum}
\begin{equation}\label{E: gdatafitting}
\begin{split}
       ||\bm{P} - \bm{W}\bm{G}||_F^2 &= ||\bm{P} - \sum_{l = 1}^K\bm{W}(:, l)\bm{G}(l, :)||_F^2
       \\& = ||\bm{P}  - \sum_{l \neq k}\bm{W}(:, l)\bm{G}(l, :) - \bm{W}(:, k)\bm{G}(k, :)||_F^2
      \\&  = ||\bm{P}_k - \bm{W}(:, k)\bm{G}(k, :)||_F^2 \\&
        = ||\bm{P}_k||_F^2 +||\bm{W}(:, k)|| _F^2||\bm{G}(k, :)|| _F^2 - 2\bm{G}(k, :)\bm{P}_k^T\bm{W}(:, k)
\end{split}
\end{equation}
\noindent where $\bm{W}(:, k)$ is a column vector with elements being the $k$th column of $\bm{W}$ and $\bm{G}(k, :)$ is a row vector with elements being the $k$th row of $\bm{G}$, and $\bm{P}_k = \bm{P}  - \sum_{l \neq k}\bm{W}(:, l)\bm{G}(l, :)$. 

Let $\bar{\bm{G}_k}$ be the submatrix of $\bm{G}$ by removing the $k$th
row of $\bm{G}$. Then $\bm{G} = \bm{\Gamma}\begin{bmatrix}
\bm{G}(k, :) \\
\bar{\bm{G}_k}
\end{bmatrix}$ where $\bm{\Gamma}$ is a permutation matrix. Thus we have
\begin{equation}\label{E: detfull}
\begin{split}
   \text{det}(\bm{G}\bm{G}^T) 
   &= \text{det}\left(\begin{bmatrix}
\bm{G}(k, :) \\
\bar{\bm{G}_k}
\end{bmatrix}\begin{bmatrix}
\bm{G}(k, :) \\
\bar{\bm{G}_k}
\end{bmatrix}^T\right) \\
&= \text{det}\left(\begin{bmatrix}
\bm{G}(k, :) \\
\bar{\bm{G}_k}
\end{bmatrix}\left[\bm{G}(k, :)^T, \bar{\bm{G}_k}^T\right]\right) \\
&= \text{det}\left(\bar{\bm{G}_k}\bar{\bm{G}_k}^T\right)\text{det}\left(\bm{G}(k, :)(\bm{I} - \bar{\bm{G}_k}^T(\bar{\bm{G}_k}\bar{\bm{G}_k}^T)^{-1}\bar{\bm{G}_k})\bm{G}(k, :)^T\right)\\
&= \text{det}\left(\bar{\bm{G}_k}\bar{\bm{G}_k}^T\right)\left(\bm{G}(k, :)(\bm{I} - \bar{\bm{G}_k}^T(\bar{\bm{G}_k}\bar{\bm{G}_k}^T)^{-1}\bar{\bm{G}_k})\bm{G}(k, :)^T\right).
\end{split}
\end{equation} 
The last equality holds because $\left(\bm{G}(k, :)(\bm{I} - \bar{\bm{G}_k}^T(\bar{\bm{G}_k}\bar{\bm{G}_k}^T)^{-1}\bar{\bm{G}_k})\bm{G}(k, :)^T\right)$ is a matrix of size $1\times 1$. We have
\begin{equation}\label{E: detpart}
\bm{G}(k, :)(\bm{I} - \bar{\bm{G}_k}^T(\bar{\bm{G}_k}\bar{\bm{G}_k}^T)^{-1}\bar{\bm{G}_k})\bm{G}(k, :)^T 
   = \bm{G}(k, :)\bm{C}_k\bm{C}_k^T\bm{G}(k, :)^T
\end{equation} 
where $\bm{C}_k = \text{Null}(\bar{\bm{G}_k})$ is an orthonormal basis for the null space of $\bar{\bm{G}_k}$. The reason can be seen as follows. From $\bm{C}_k = \text{Null}(\bar{\bm{G}_k})$, we have
$\bm{C}_k^T\bm{C}_k = \bm{I}$ and $\bar{\bm{G}_k}\bm{C}_k = \bm{0}$. The column of $\bm{C}_k$ and the rows of $\bar{\bm{G}_k}$ together form a base of the $J$-dimensional space. Therefore, any vector $\bm{G}(k, :)^T$ can be expressed as $\bm{G}(k, :)^T = \bm{C}_k\bm{x} + \bar{\bm{G}_k}^T\bm{y}$. Then we have 
\begin{equation*}
    \bm{G}(k, :)\bm{G}(k, :)^T = \bm{x}^T\bm{x} + \bm{y}^T\bar{\bm{G}_k}\bar{\bm{G}_k}^T\bm{y}.
\end{equation*}
Note that $\bar{\bm{G}_k}\bm{G}(k, :)^T = \bar{\bm{G}_k}\bar{\bm{G}_k}^T\bm{y}$. Therefore, $\bm{G}(k, :)\bar{\bm{G}_k}^T(\bar{\bm{G}_k}\bar{\bm{G}_k}^T)^{-1}\bar{\bm{G}_k}\bm{G}(k, :)^T = \bm{y}^T\bar{\bm{G}_k}\bar{\bm{G}_k}^T\bm{y}$.
Note that $\bm{C}_k^T\bm{G}(k, :)^T = \bm{x}$. Hence, $\bm{x}^T\bm{x} = \bm{G}(k, :)\bm{C}_k\bm{C}_k^T\bm{G}(k, :)^T$. Therefore, 
\begin{equation*}
    \bm{G}(k, :)\bm{G}(k, :)^T - \bm{G}(k, :)\bar{\bm{G}_k}^T(\bar{\bm{G}_k}\bar{\bm{G}_k}^T)^{-1}\bar{\bm{G}_k}\bm{G}(k, :)^T = \bm{x}^T\bm{x} = \bm{G}(k, :)\bm{C}_k\bm{C}_k^T\bm{G}(k, :)^T.
\end{equation*}
The proof for Equation~(\ref{E: detpart}) is completed. Combined Equation~(\ref{E: detfull}) with Equation~(\ref{E: detpart}), volume term $\text{det}(\bm{G}\bm{G}^T)$ can be expressed as
\begin{equation}\label{E: det}
    \text{det}(\bm{G}\bm{G}^T) = \text{det}(\bar{\bm{G}_k} \bar{\bm{G}_k}^T)\bm{G}(k, :)\bm{C}_k\bm{C}_k^T\bm{G}(k, :)^T
\end{equation}

Combined Equation~(\ref{E: gdatafitting}) with Equation~(\ref{E: det}), the problem of Equation~(\ref{Eq: gl1normalter}) is decomposed into
$K$ independent sub-problems which are quadratic programming problems: 
\begin{equation}\label{eq: eachg}
    \begin{split}
      \text{min}_{\bm{G}(k, :)}~~~& \bm{G}(k, :)\bm{Q}_k\bm{G}(k, :)^T - \bm{W}(:, k)^T\bm{P}_k\bm{G}(k, :)^T
        \\
        \text{subject to }~~~& \bm{G}(k, :)\bm{1} = \bm{1}, \bm{G}(k, :) \geq 0
    \end{split}
\end{equation}
\noindent where $\bm{Q}_k = \frac{1}{2}\bm{W}(:, k)^T\bm{W}(:, k)\bm{I} - \frac{\lambda}{2} \text{det}(\bar{\bm{G}_k}\bar{\bm{G}_k}^T)\bm{C}_k\bm{C}_k^T$ and $\bm{P}_k = \bm{P}  - \sum_{l \neq k}\bm{W}(:, l)\bm{G}(l, :)$. In practice, the subproblems are solved alternatively where $\bm{Q}_k$ and $\bm{P}_k$ use the latest rows of $\bm{G}$. The strict convexity of Equation~(\ref{eq: eachg}) requires $\bm{Q}_k$ to be positive definite. Consequently, the value of $\lambda$ cannot be too large. Again, following \citet{Leplat2019rank, GillisSNPA2014}, we use PFGM on Equation~(\ref{eq: eachg}) \citep{nesterov2004introductory}.

We call this algorithm APFGM, short for alternative projected fast gradient methods. See algorithm \ref{alg: APFGM} for a conclusion, where $\lambda$ in (\ref{E: emamodnoise}) is computed by scaling the input $\lambda'$: $\lambda = \lambda'\frac{||\bm{P} - \bm{W}^{(0)}\bm{G}^{(0)}||_F^2}{\text{det}(\bm{G}^{(0)}(\bm{G}^{(0)})^T)}$, with $\bm{W}^{(0)}$ and $\bm{G}^{(0)}$ being the initial input matrices.

Note that, in APFGM, $\lambda$ can be negative, where APFGM tends to obtain basis vectors as closely as possible, which follows the minimum volume assumption \citep{Zhou2012Minimum}.

\begin{algorithm}[h]
\caption{Alternative projected fast gradient methods (APFGM)}
\label{alg: APFGM}
\KwIn{Input nonnegative matrix $P \in \mathbb{R}_{+}^{I \times J}$ with row-sum-to-one constraint $\bm{P1} = \bm{1}$, dimensionality $K$, number of iterations maxiter, and $\lambda'$.}
\KwOut{$\bm{W} \ge 0$ with $\bm{W1} = \bm{1}$ and $\bm{G} \ge 0$ with $\bm{G1} = \bm{1}$.}

Generate initial matrices $\bm{W}^{(0)} \ge 0$  with $\bm{W}^{(0)}\bm{1} = \bm{1}$ and $\bm{G}^{(0)} \ge 0$ with $\bm{G}^{(0)}\bm{1} = \bm{1}$. Let $\lambda = \lambda'\frac{||\bm{P} - \bm{W}^{(0)}\bm{G}^{(0)}||_F^2}{\text{det}(\bm{G}^{(0)}(\bm{G}^{(0)})^T)}$.

\For{$t = 1, 2, \ldots, $\text{maxiter}}{
  Apply FPGD to the quadratic programming problem Equation~(\ref{eq: eachw}) to update abundance matrix $\bm{W}$
  
  Apply FPGD to the quadratic programming problem Equation~(\ref{eq: eachg}) to update end-member matrix $\bm{G}$
}
\end{algorithm}

\subsection{Applying APFGM to the  simulated GSD data  in the Introduction}
We perform APFGM on the same simulated GSD data as in the Introduction. When $\lambda < 0$, i.e., minimum volume constraint, the results in Figures~\ref{F: QPtwosourceminendmem} and \ref{F: QPtwosourceminabundances} are similar to Figures~\ref{F: EMMAtwosourcemaxendmem} and ~\ref{F: EMMAtwosourcemaxabundances} from EMMA, and cannot obtain the true EMs and misestimate abundances. We find the same results in Figures~\ref{F: QPtwosourcenoendmem} and \ref{F: QPtwosourcenoabundances} for $\lambda = 0$, i.e., no volume constraint. However, when $\lambda > 0$, i.e., using the maximum volume assumption, APFGM recovers the true EMs and abundances; see Figures~\ref{F: QPtwosourcemaxendmem} and \ref{F: QPtwosourcemaxabundances}. The determinant of $\bm{G}\bm{G}^T$ for APFGM with minimum volume, APFGM with no volume, and EMMA is around 0.0015, but for APFGM with maximum volume is around 0.0028. APFGM with maximum volume tends to find end members as different as possible which makes it suitable for this highly mixed GSD data.

\begin{figure}[h]
\centering
\begin{subfigure}[b]{0.32\linewidth}
\includegraphics[width=1\textwidth]{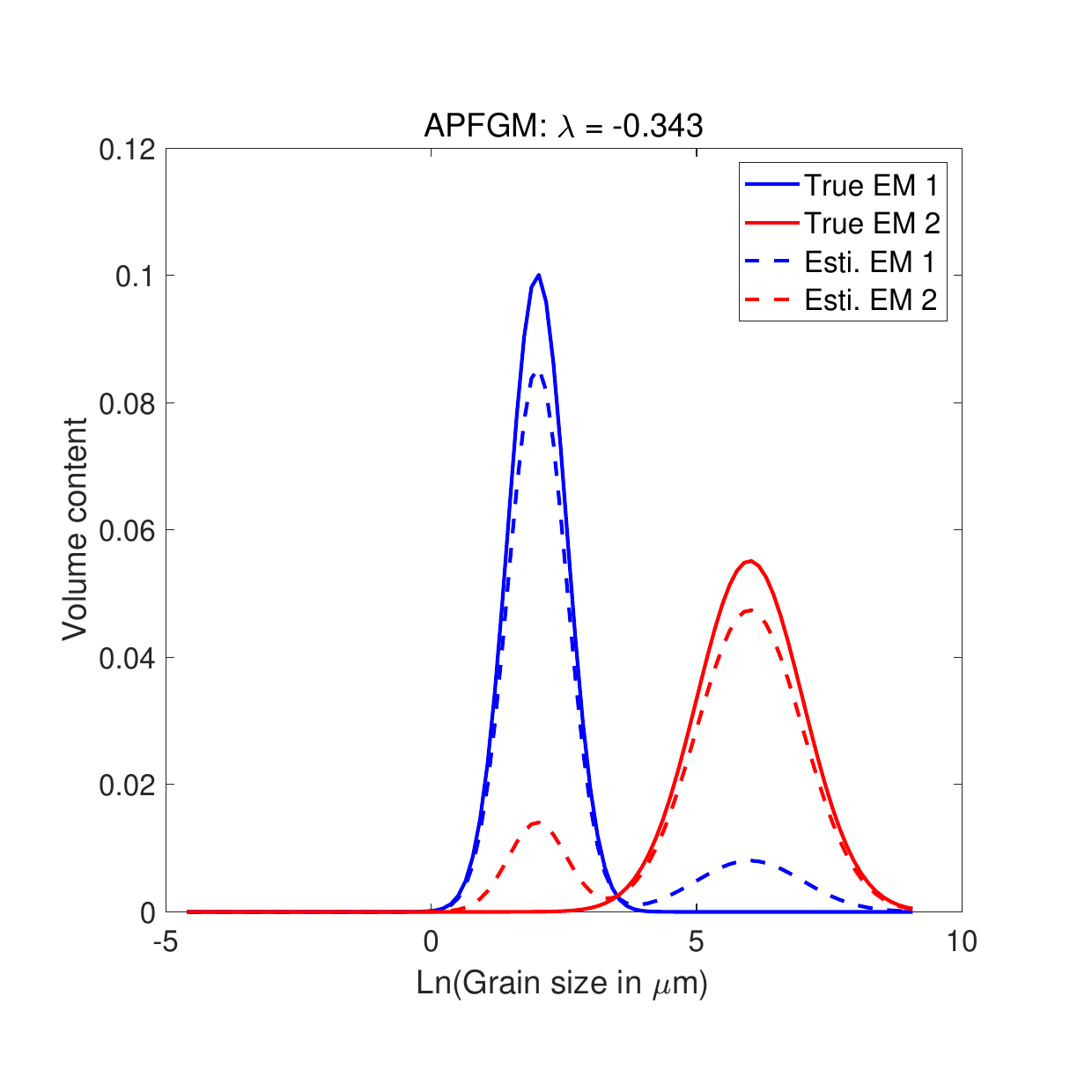}
    \caption{$\lambda = -0.343$: End Members}
    \label{F: QPtwosourceminendmem}
 \end{subfigure}
    \begin{subfigure}[b]{0.32\linewidth}
\includegraphics[width=1\textwidth]{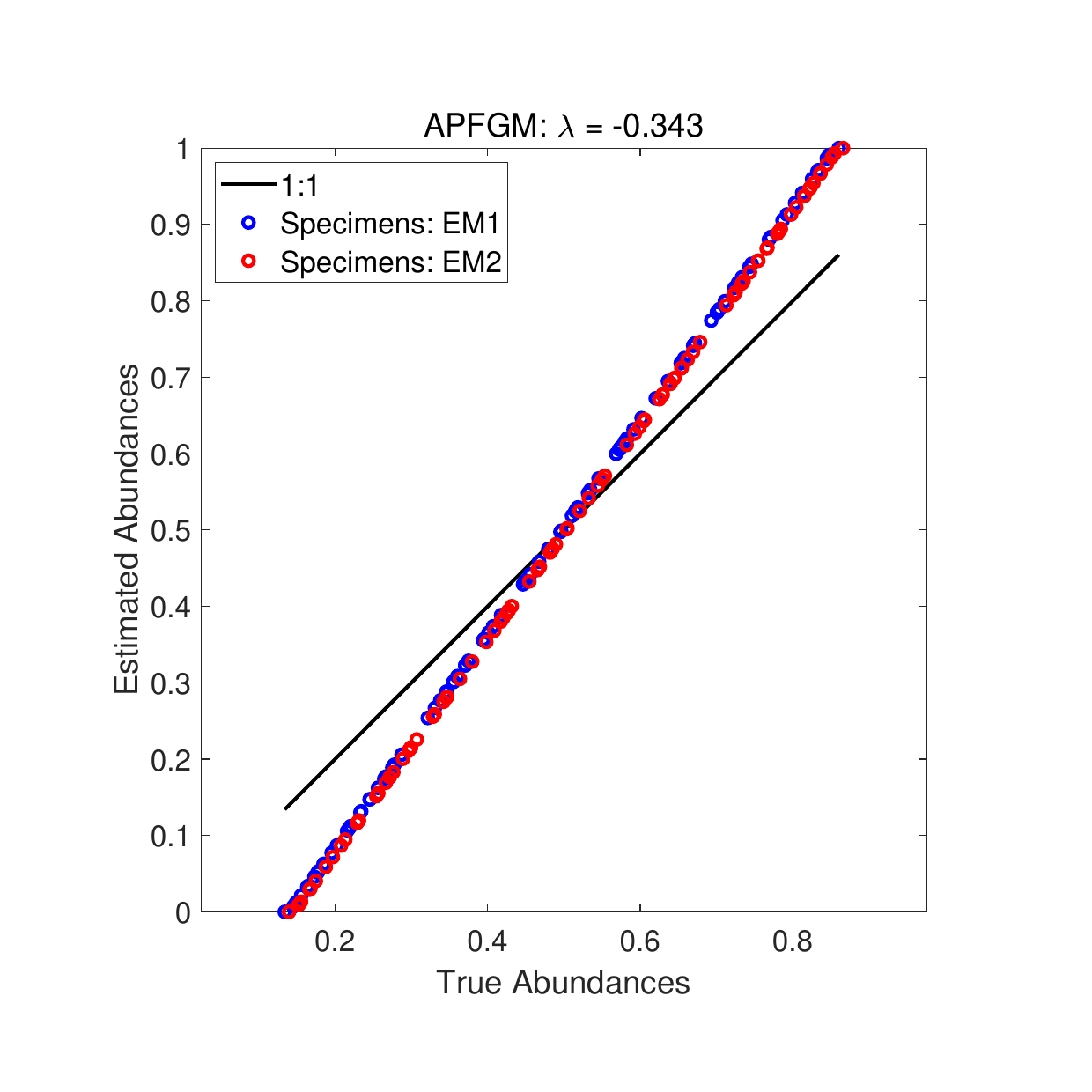}
    \caption{$\lambda = -0.343$: Abundances}
    \label{F: QPtwosourceminabundances}
   \end{subfigure} 
 \begin{subfigure}[b]{0.32\linewidth}
\includegraphics[width=1\textwidth]{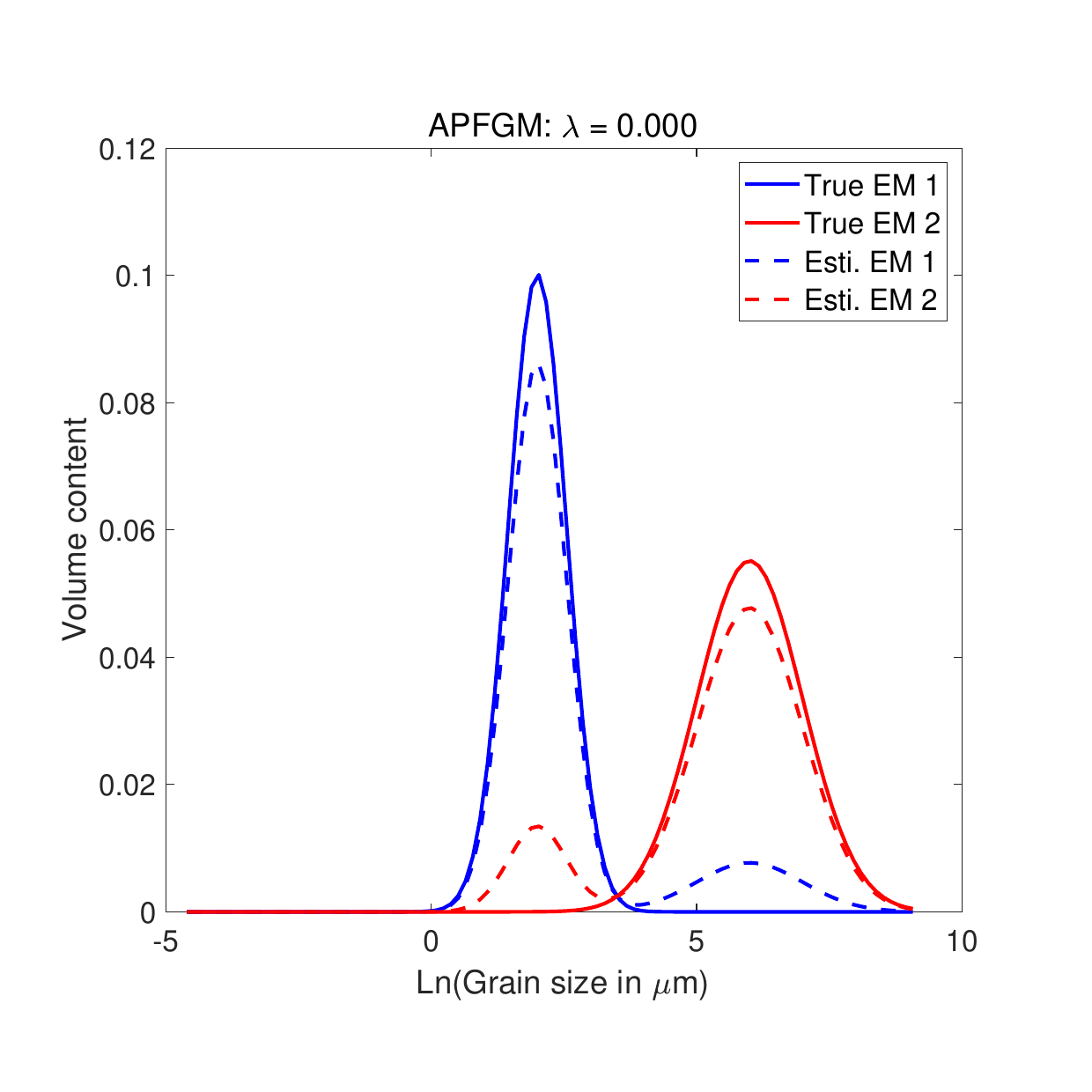}
    \caption{$\lambda = 0$: End Members}
    \label{F: QPtwosourcenoendmem}
 \end{subfigure}
    \begin{subfigure}[b]{0.32\linewidth}
\includegraphics[width=1\textwidth]{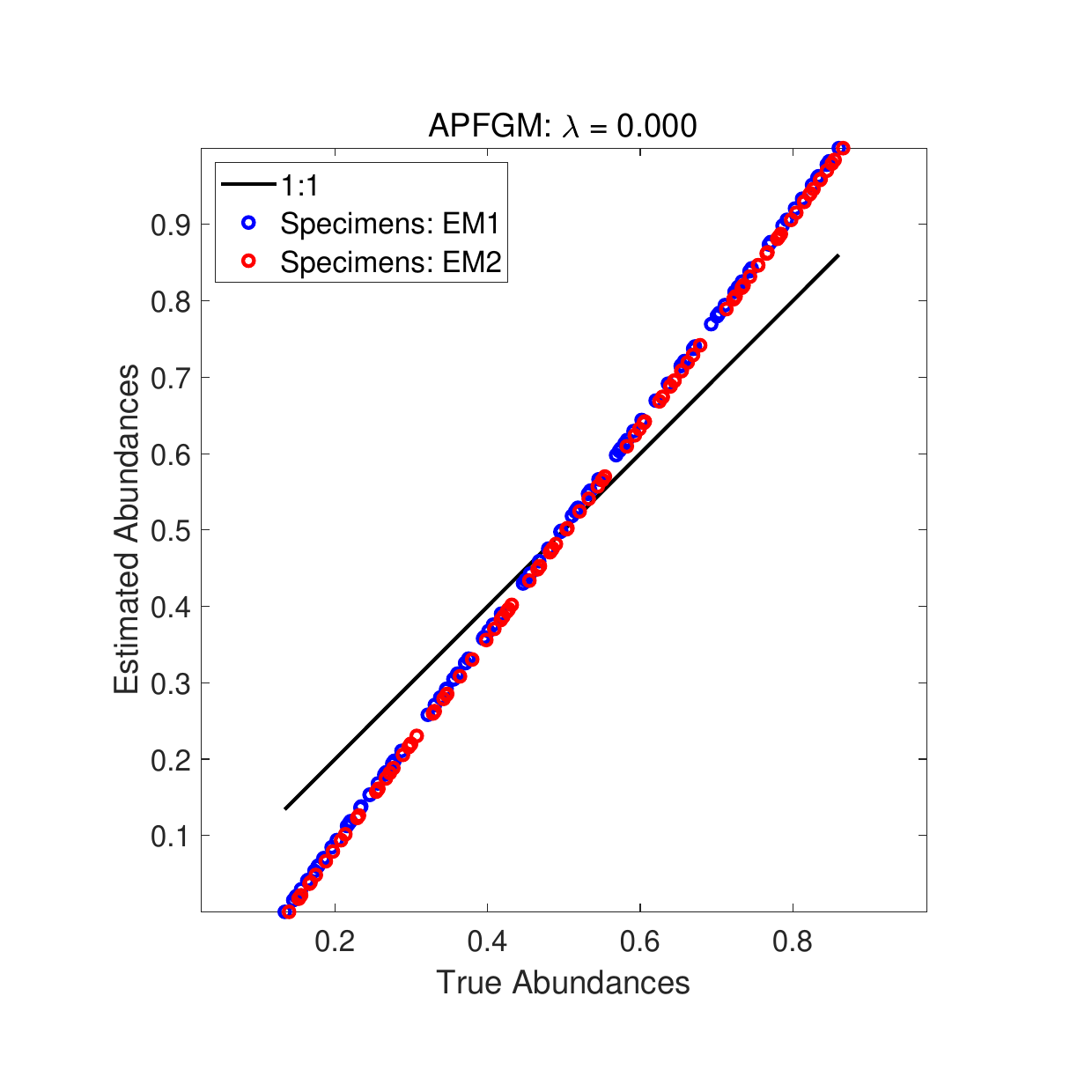}
    \caption{$\lambda = 0$: Abundances}
    \label{F: QPtwosourcenoabundances}
 \end{subfigure}
 \begin{subfigure}[b]{0.32\linewidth}
\includegraphics[width=1\textwidth]{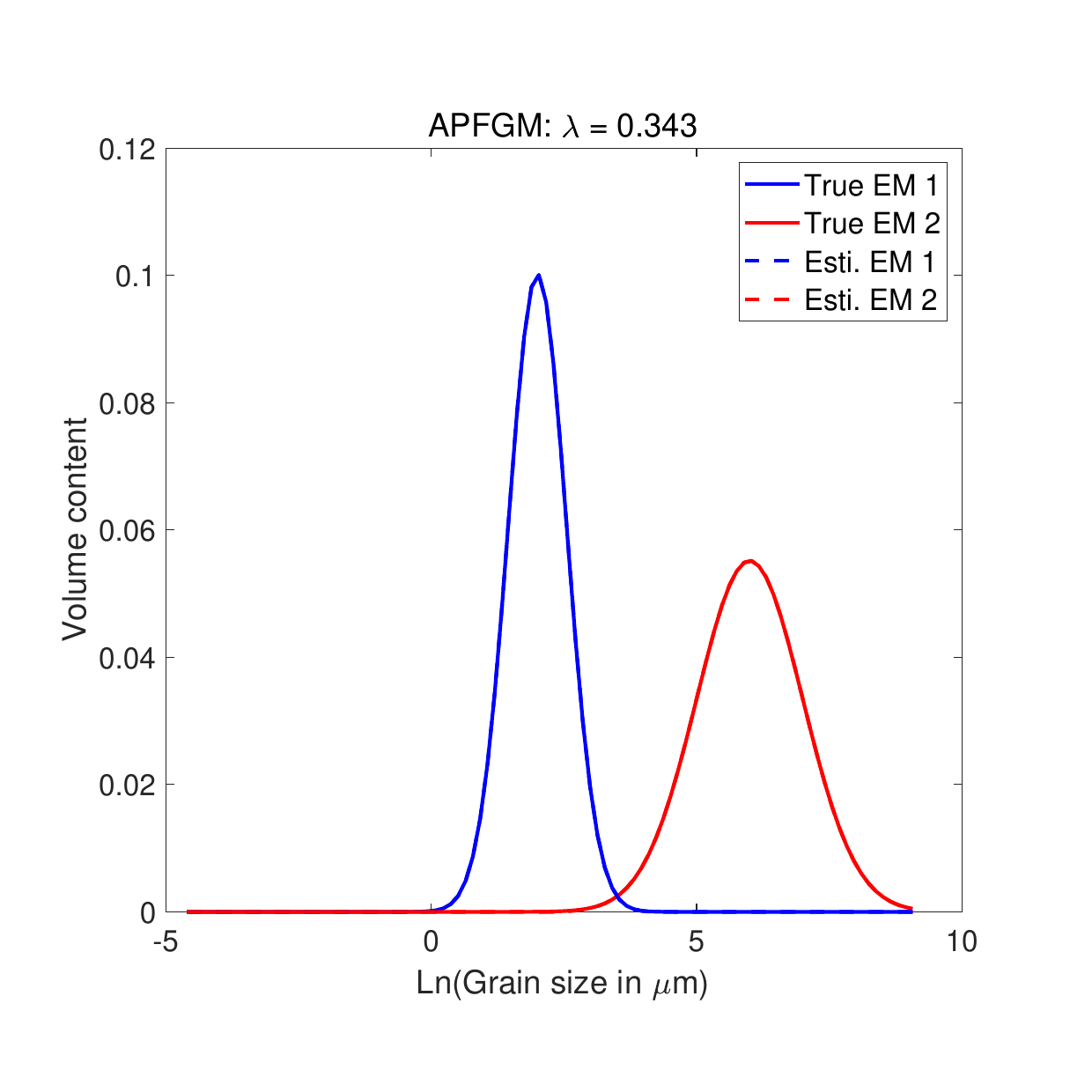}
    \caption{$\lambda = 0.343$: End Members}
    \label{F: QPtwosourcemaxendmem}
 \end{subfigure}
    \begin{subfigure}[b]{0.32\linewidth}
\includegraphics[width=1\textwidth]{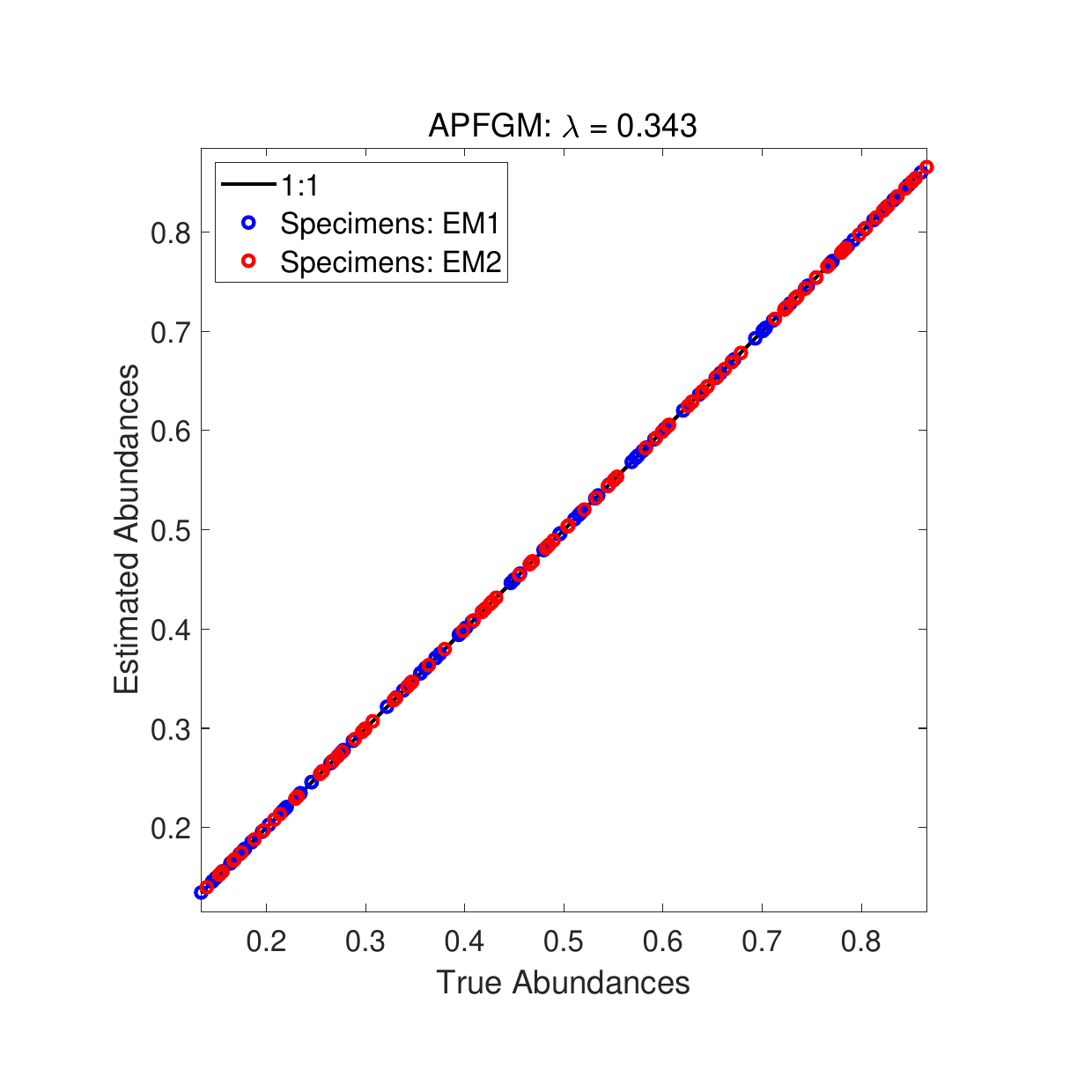}
    \caption{$\lambda = 0.343$: Abundances}
    \label{F: QPtwosourcemaxabundances}
 \end{subfigure}
 \caption{APFGM:
(a and b) $\lambda = -0.343$; (c and d) $\lambda = 0$; (e and f) $\lambda = 0.343$.}\label{F: maximumnoisycoversanddataset}
    \end{figure}

\section{Experiments}\label{S: experiments}

In this section, we simulate different levels of mixing GSD data to extensively explore the performance of APFGM. The code for this paper is available on the Github website https://github.com/qianqianqi28/MVC-EMA, implemented in MATLAB R2024b and R 4.2.3.

\subsection{Generation of artificial GSD data}

\begin{figure}[h]
\centering
\includegraphics[width=0.55\textwidth]{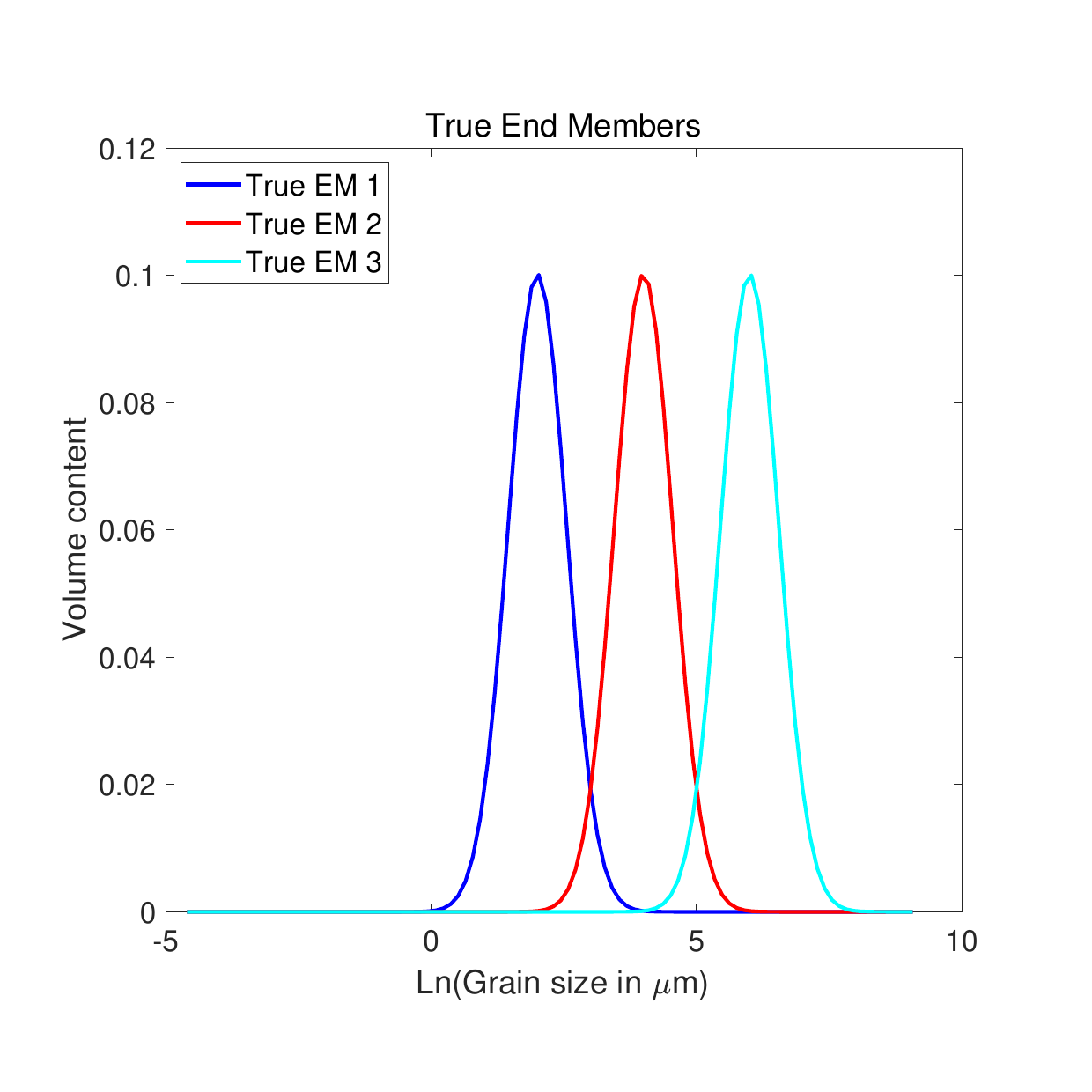}
\caption{Three lognormal EMs.}\label{F: True End Members three sources}
\end{figure}

Three end members (EMs), simulated by lognormal distributions, are presented in Figure~\ref{F: True End Members three sources} where each EMs have 100 values. Six artificial GSD data are then produced by combining these EMs with random abundances where each GDS data have 200 specimen, and the minimum abundance for each specimen was varied from 0 to 0.25 in increments of 0.05 to simulate conditions ranging from poorly mixed to highly mixed, respectively.

\subsection{Indexes used to evaluate the unmixing results}

We use mean angular deviations (in degrees) between true EMs and estimated EMs to evaluate estimated GSDs:

\begin{equation}
   \text{MAEM} = \frac{180}{\pi K} \sum_{k = 1}^K \arccos\left( \frac{\langle \text{True }\bm{G}(k, :), \text{Est. } \bm{G}(k, :) \rangle}{\|\text{True }\bm{G}(k, :)\|\|\text{Est. }\bm{G}(k, :)\|} \right) 
\end{equation}
and use mean angular deviations (in degrees) between true abundances and estimated abundances to evaluate estimated abundances:
\begin{equation}
   \text{MAAB} = \frac{180}{\pi I} \sum_{i = 1}^I \arccos\left( \frac{\langle \text{True }\bm{W}(i, :), \text{Est. } \bm{W}(i, :) \rangle}{\|\text{True }\bm{W}(i, :)\|\|\text{Est. }\bm{W}(i, :)\|} \right) 
\end{equation}

The range of MAEM and MAAB is larger or equal to 0 degrees and less or equal to 90 degrees. The smaller MAEM or EAAB is, the better an algorithm is. MAEM (MAAB) with a value of 0 means 
that estimated EMs (abundances) are identical to true EMs (abundances). MAEM (MAAB) with a value of 90 means 
that estimated EMs (abundances) are orthogonal to true EMs (abundances).

\subsection{Comparison of EMMA,
APFGM with minimum volume, APFGM with no volume, and APFGM with maximum volume}

The scaling factor $\lambda'$ controls the volume regularization in APFGM. It is set to $-1$, $0$, $1$ for APFGM with minimum volume, with no volume, and with maximum volume, respectively. The MAEM and MAAB against increasing degree
of mixing are plotted in Figure~\ref{F: End Member Misfit} and Figure~\ref{F: Abundance Misfit} respectively and are shown in Table~\ref{T: End Member Misfit} and Table~\ref{T: Abundance Misfit} respectively. We can see that APFGM with maximum volume is best, then APFGM with no volume, finally, EMMA and APFGM with minimum volume. Specifically, at zero level, APFGM with no volume and APFGM with maximum volume estimate
the true end-members and abundances more accurately than EMMA and APFGM with minimum volume. This means that APFGM algorithm is also suitable for no mixed GSD data. As level of mixing increases, EMMA and APFGM with no volume and minimum volume cannot clearly estimate EMs and abundances, but APFGM with maximum volume
fits the true EMs and abundances, with a slight deterioration. This highlights the usefulness of
maximum volume regularizer in highly mixed data sets. 

\begin{figure}[H]
\centering
 \begin{subfigure}[b]{0.45\linewidth}
\includegraphics[width=1\textwidth]{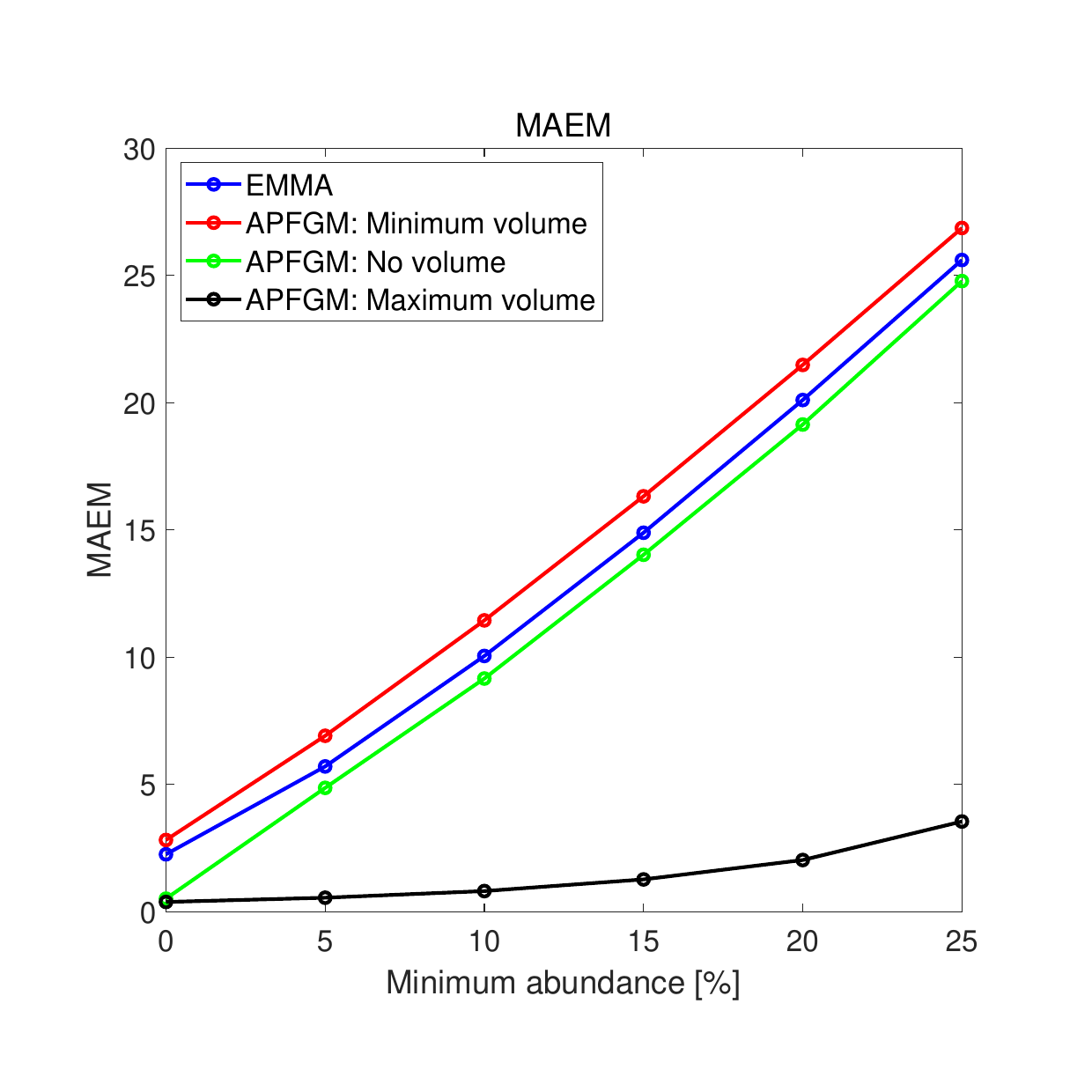}
    \caption{MAEM}
    \label{F: End Member Misfit}
 \end{subfigure}
    \begin{subfigure}[b]{0.45\linewidth}
\includegraphics[width=1\textwidth]{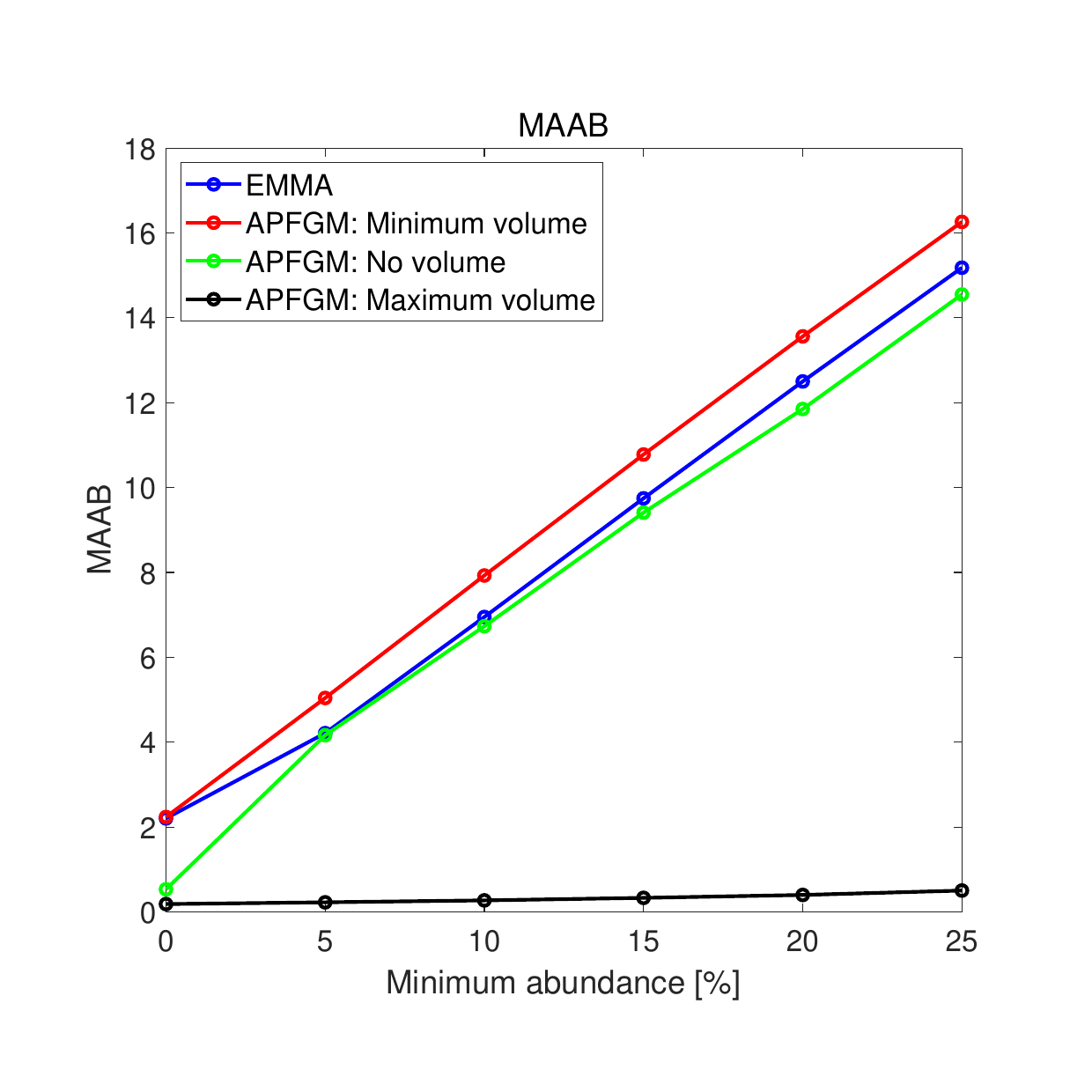}
    \caption{MAAB}
    \label{F: Abundance Misfit}
 \end{subfigure}
 \caption{EMMA, APFGM with minimum volume, APFGM with no volume, and
APFGM with maximum volume: (a) MAEM; (b) MAAB.}\label{F: maemmaab}
    \end{figure}

\begin{table}[H]
    \caption{The MAEM of EMMA, APFGM with minimum volume, APFGM with no volume, and
APFGM with maximum volume under different mixed level.}
    \label{T: End Member Misfit}
    \centering
    \begin{tabular}{ccccc}
    \hline
    Mixed level & EMMA & APFGM min & APFGM no &  APFGM max
    \\
    \hline
 0.00&2.2602 & 2.8223&  0.5175& 0.3883\\
0.05& 5.7117& 6.9162& 4.8707& 0.5551\\
0.10&10.0522&11.4498& 9.1652& 0.8163\\
0.15&14.8928&16.3239&14.0272& 1.2743\\
0.20&20.1057&21.4845&19.1438& 2.0358\\
0.25&25.6058&26.8632&24.7793& 3.5488
 \\
 \hline
    \end{tabular}
\end{table}

\begin{table}[H]
    \caption{The MAAB of EMMA, APFGM with minimum volume, APFGM with no volume, and
APFGM with maximum volume under different mixed level.}
    \label{T: Abundance Misfit}
    \centering
    \begin{tabular}{ccccc}
    \hline
    Mixed level & EMMA & APFGM min & APFGM no &  APFGM max
    \\
    \hline
0.00&2.2020&  2.2359& 0.5307& 0.1834\\
0.05& 4.2121& 5.0408& 4.1633& 0.2248\\
0.10& 6.9495& 7.9277& 6.7321& 0.2700\\
0.15& 9.7462&10.7792& 9.4098& 0.3296\\
0.20&12.5014&13.5635&11.8517& 0.3982\\
0.25&15.1823&16.2625&14.5505& 0.5017
 \\
 \hline
    \end{tabular}
\end{table}

Figure~\ref{F: mixed15} shows the end-members and abundances for a minimum abundance of 15\%. The estimated
end-members from EMMA and APFGM with $\lambda \leq 0$ exhibit contamination with each other of three EMs and leading to mis-estimated abundances, but APFGM with maximum volume better recover the unimodal sources. Thus, APFGM with maximum volume exhibits a resistant for highly mixed GSD data.

\begin{figure}[H]
\centering
\begin{subfigure}[b]{0.32\linewidth}
\includegraphics[width=1\textwidth]{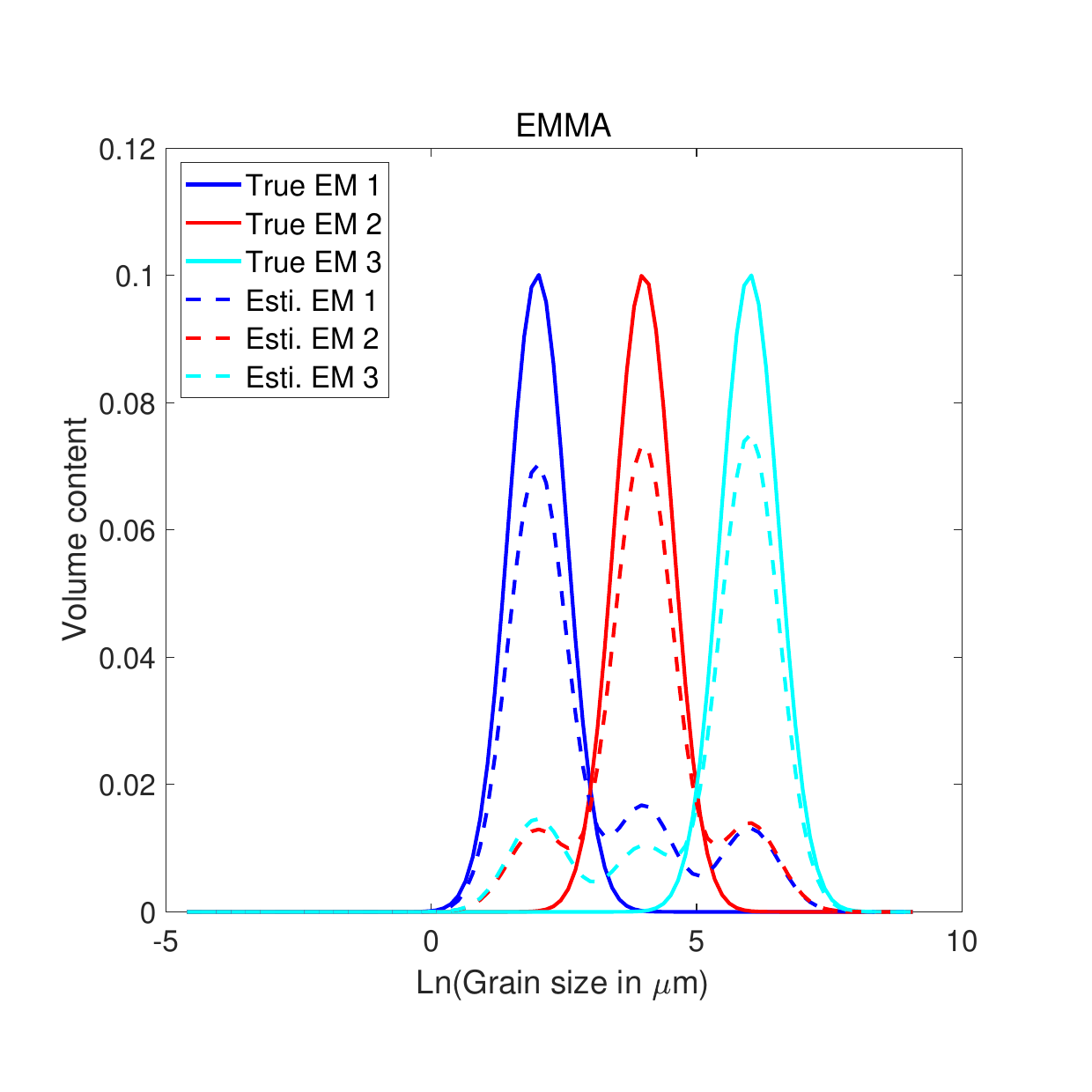}
 \caption{EMMA: EMs}\label{F: EMMAthreesourceendmem15}
 \end{subfigure}
 \begin{subfigure}[b]{0.32\linewidth}
\includegraphics[width=1\textwidth]{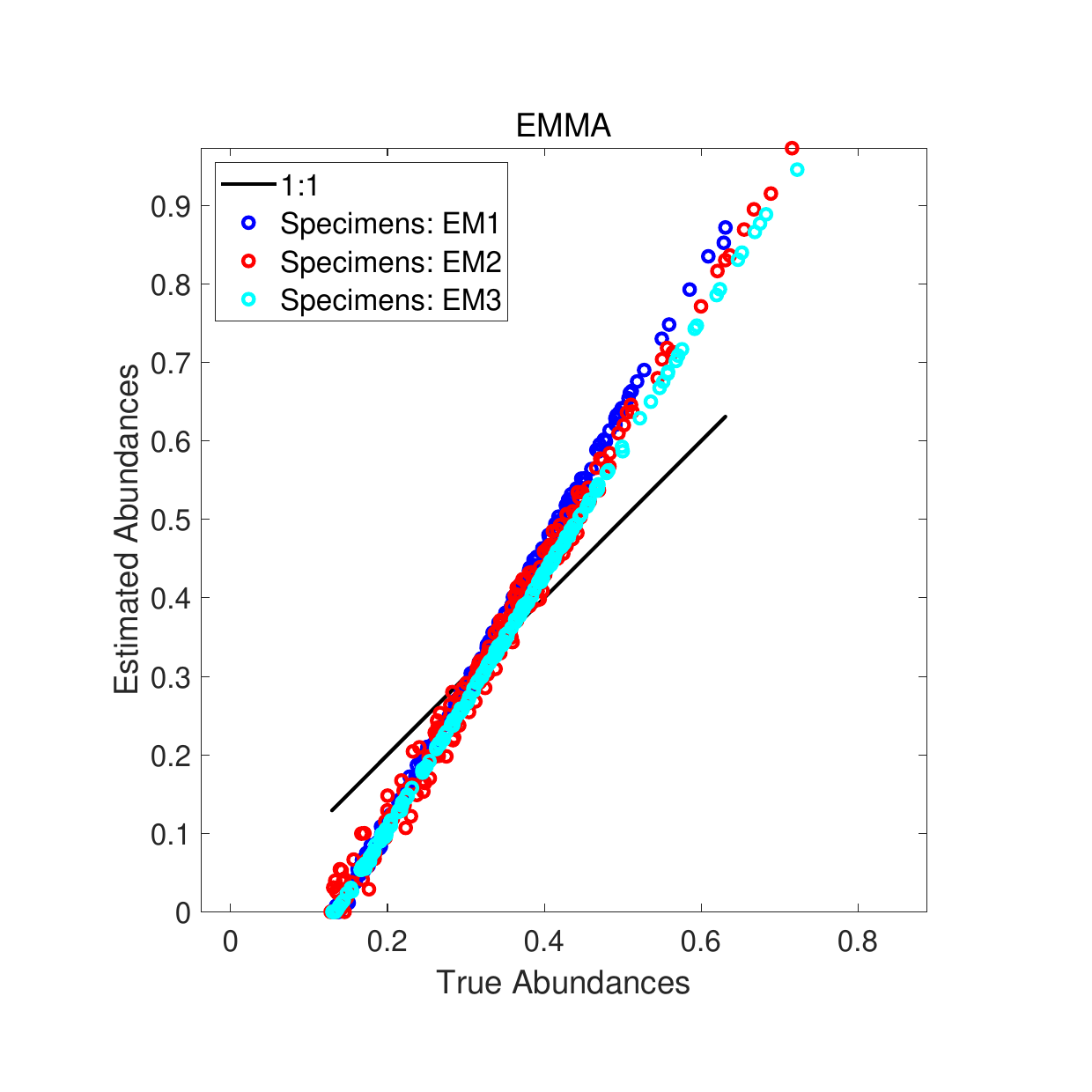}
 \caption{EMMA: Abundances}\label{F: EMMAthreesourceabundances15}
 \end{subfigure}
    \begin{subfigure}[b]{0.32\linewidth}
\includegraphics[width=1\textwidth]{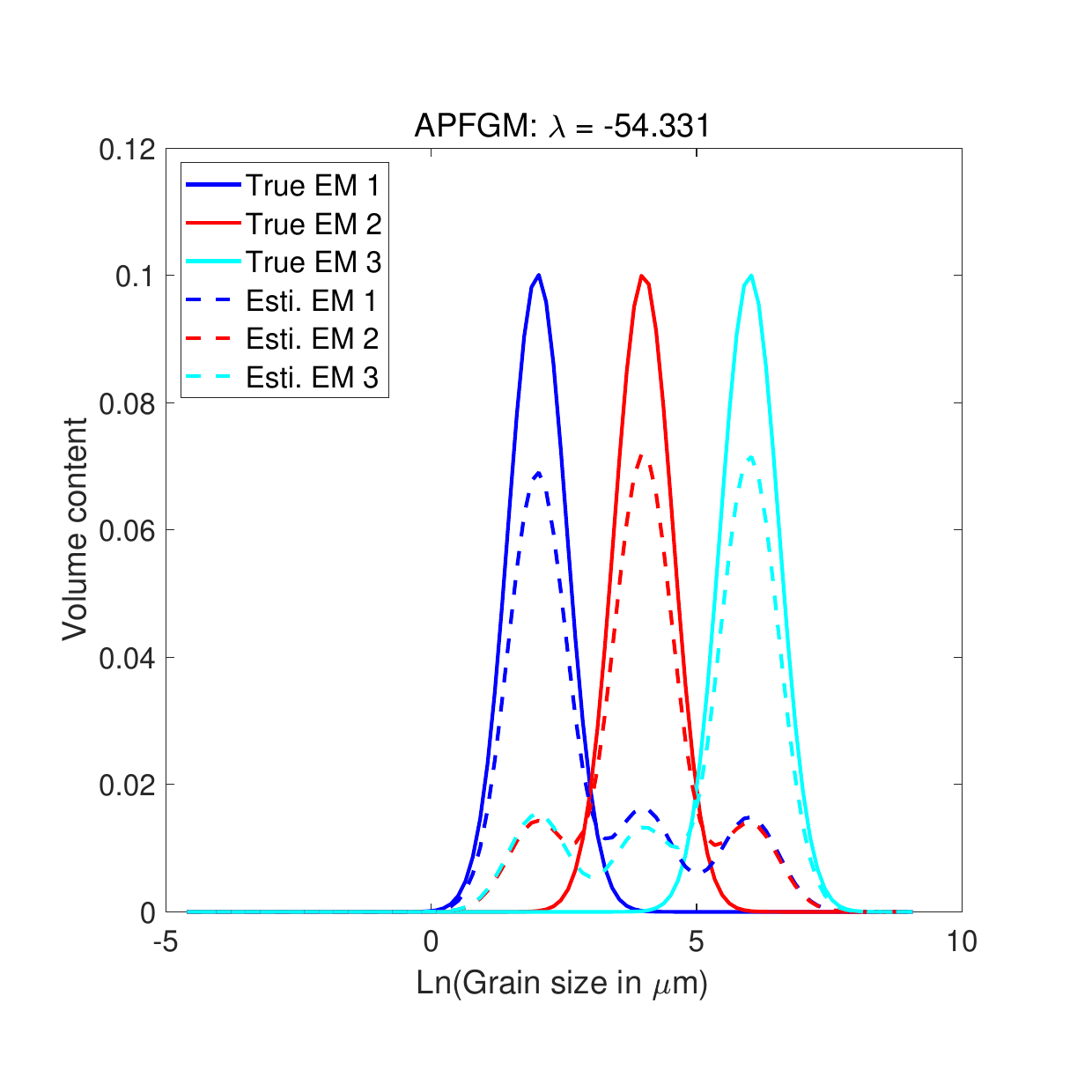}
    \caption{\tiny{PFGM with minimum volume: EMs}}
    \label{F: QPthreesourceminendmem15}
 \end{subfigure}
  \begin{subfigure}[b]{0.32\linewidth}
\includegraphics[width=1\textwidth]{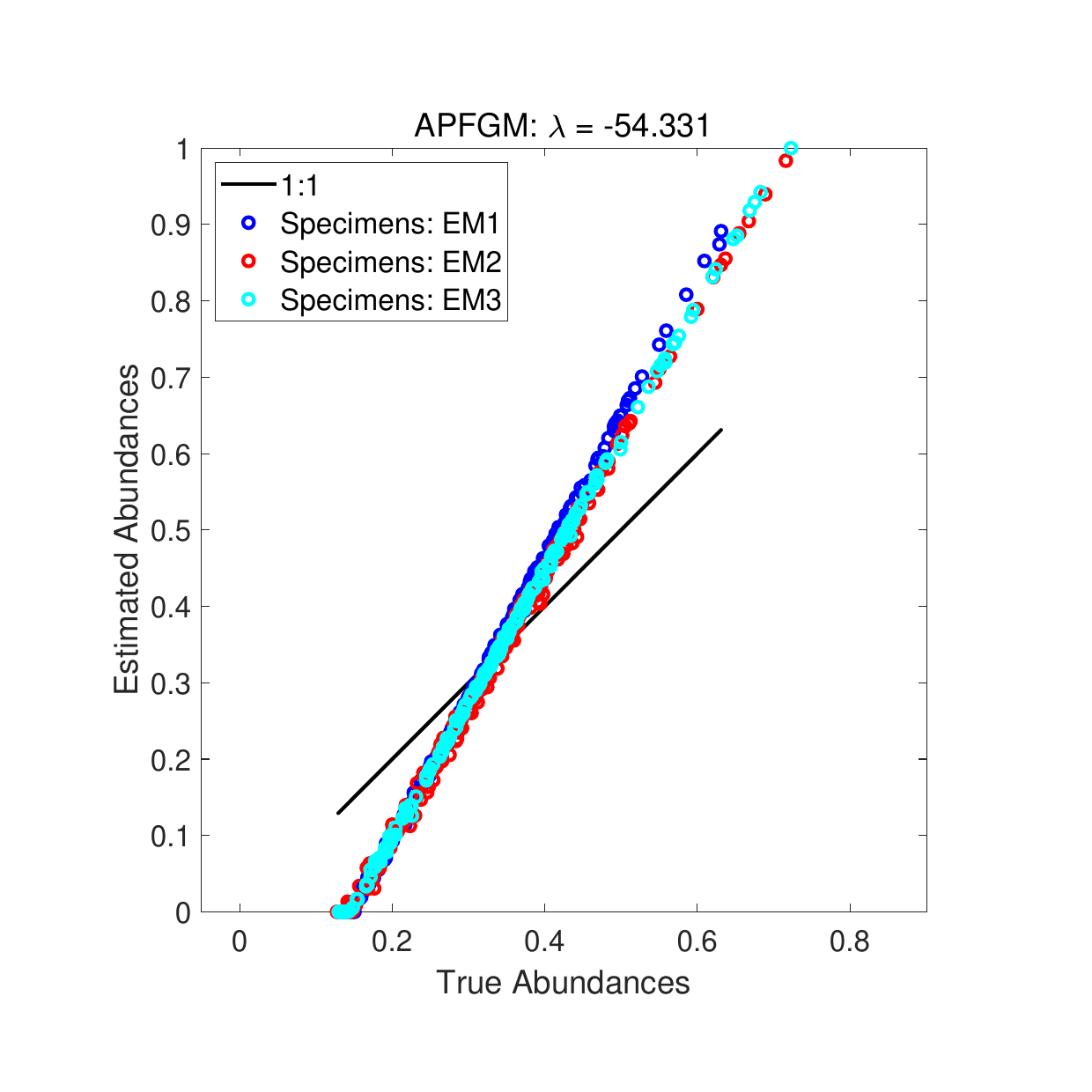}
 \caption{\tiny{PFGM with minimum volume: Abunds}}\label{F: QPthreesourceminabundances15}
 \end{subfigure}
   \begin{subfigure}[b]{0.32\linewidth}
\includegraphics[width=1\textwidth]{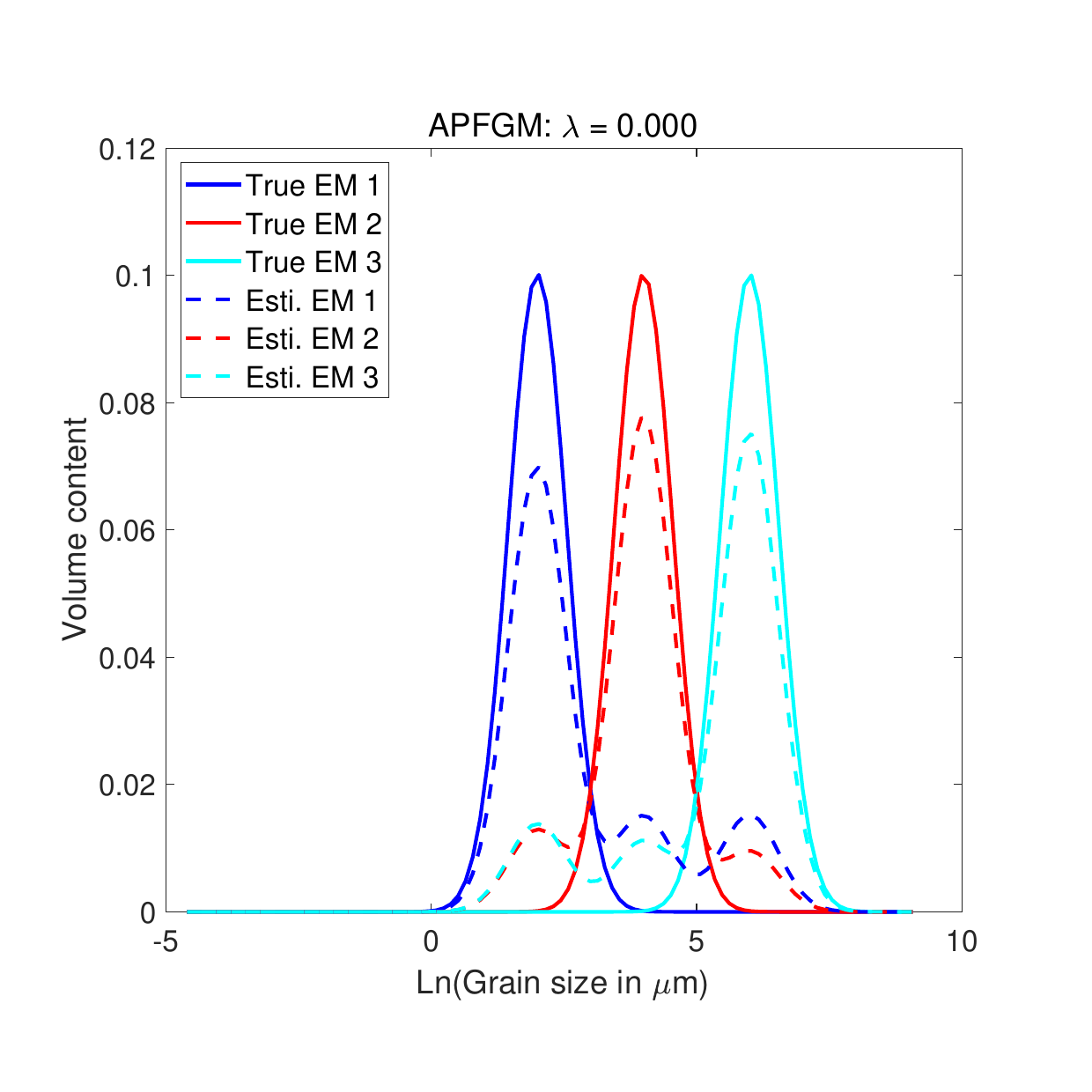}
    \caption{\tiny{PFGM with no volume: EMs}}
    \label{F: QPthreesourcenoendmem15}
 \end{subfigure}
  \begin{subfigure}[b]{0.32\linewidth}
\includegraphics[width=1\textwidth]{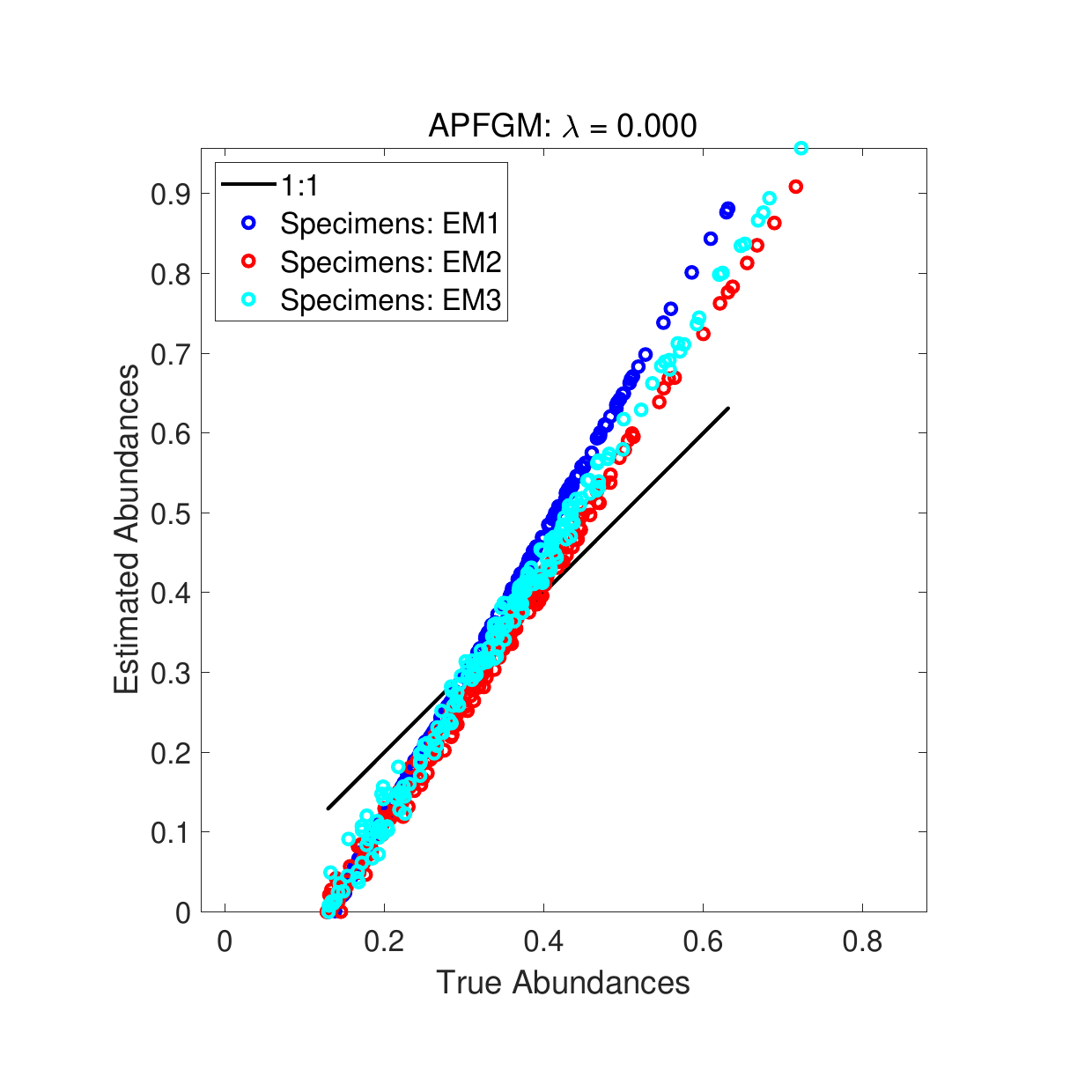}
 \caption{\tiny{PFGM with no volume: Abunds}}\label{F: QPthreesourcenoabundances15}
 \end{subfigure}
    \begin{subfigure}[b]{0.32\linewidth}
\includegraphics[width=1\textwidth]{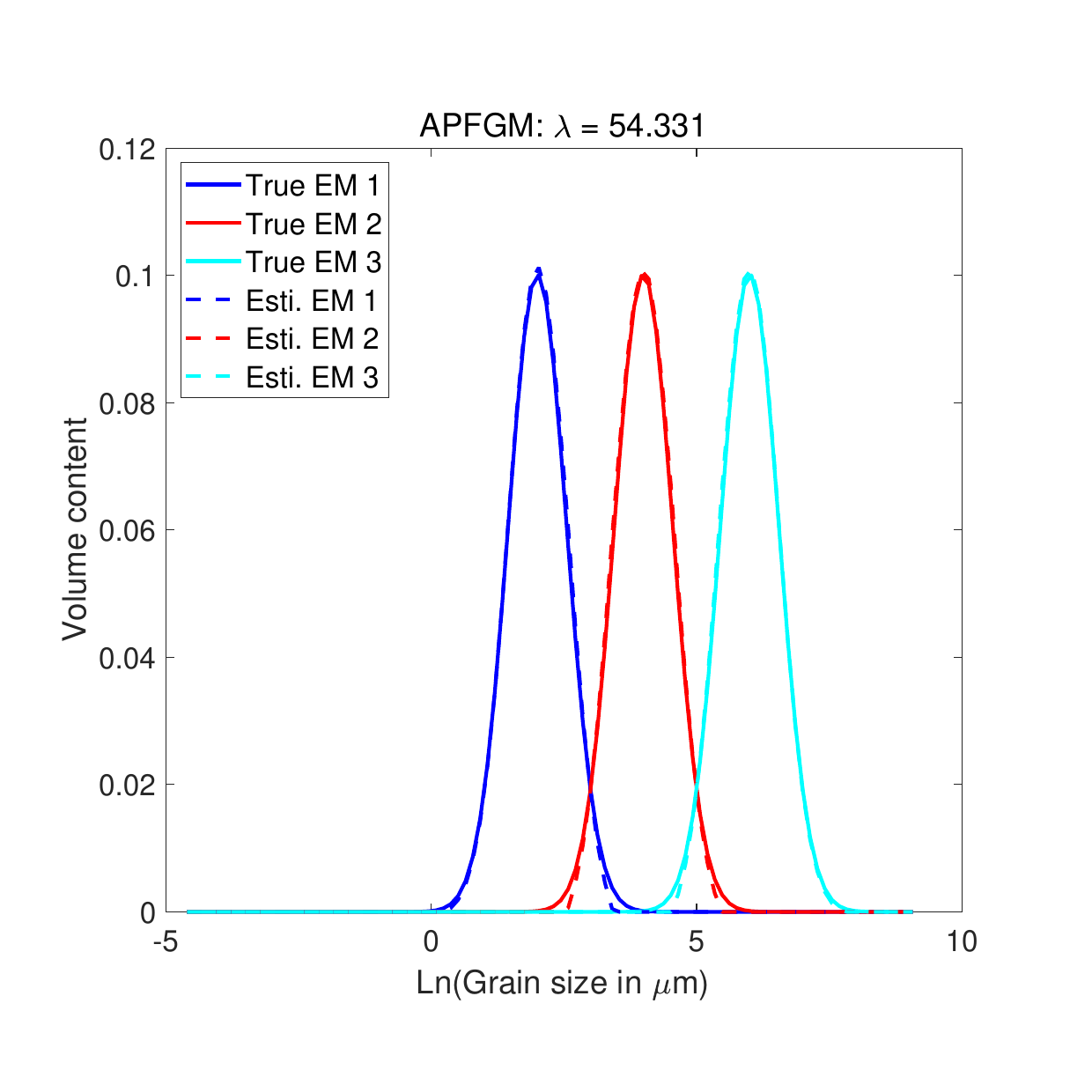}
    \caption{\tiny{PFGM with maximum volume: EMs}}
    \label{F: QPthreesourcemaxendmem15}
 \end{subfigure}
  \begin{subfigure}[b]{0.32\linewidth}
\includegraphics[width=1\textwidth]{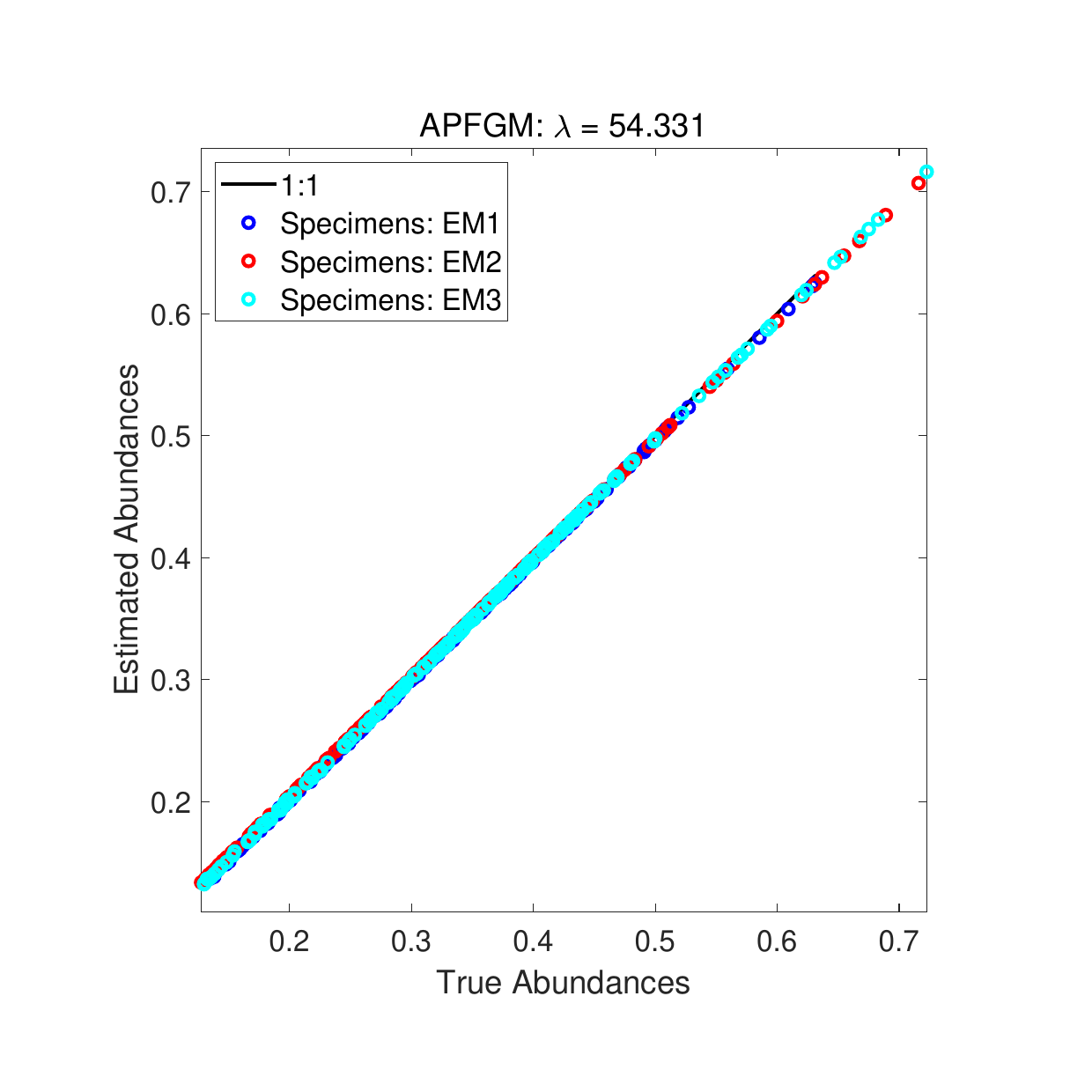}
 \caption{\tiny{PFGM with maximum volume: Abunds}}\label{F: QPthreesourcemaxabundances15}
 \end{subfigure}
 \caption{(a and b) EMMA; (c and d) PFGM with minimum volume; (e and f) PFGM with no volume; (g and h) PFGM with maximum volume.}\label{F: mixed15}
    \end{figure}

The volume $\text{det}(\bm{G}\bm{G}^T)$ against
increasing degree of mixing is shown in Figure~\ref{F: Determinant for Basis Matrix} and Table~\ref{T: Determinant for Basis Matrix}. We can see that APFGM with maximum volume has maximum volume, APFGM with no volume is second, and finally, EMMA and APFGM with minimum volume for any level of mixing. As increasing the level of mixing, the volume $\text{det}(\bm{G}\bm{G}^T)$ increases for APFGM with maximum volume and decreases for EMMA and APFGM with no volume and minimum volume.

\begin{figure}[H]
\centering
\includegraphics[width=0.55\textwidth]{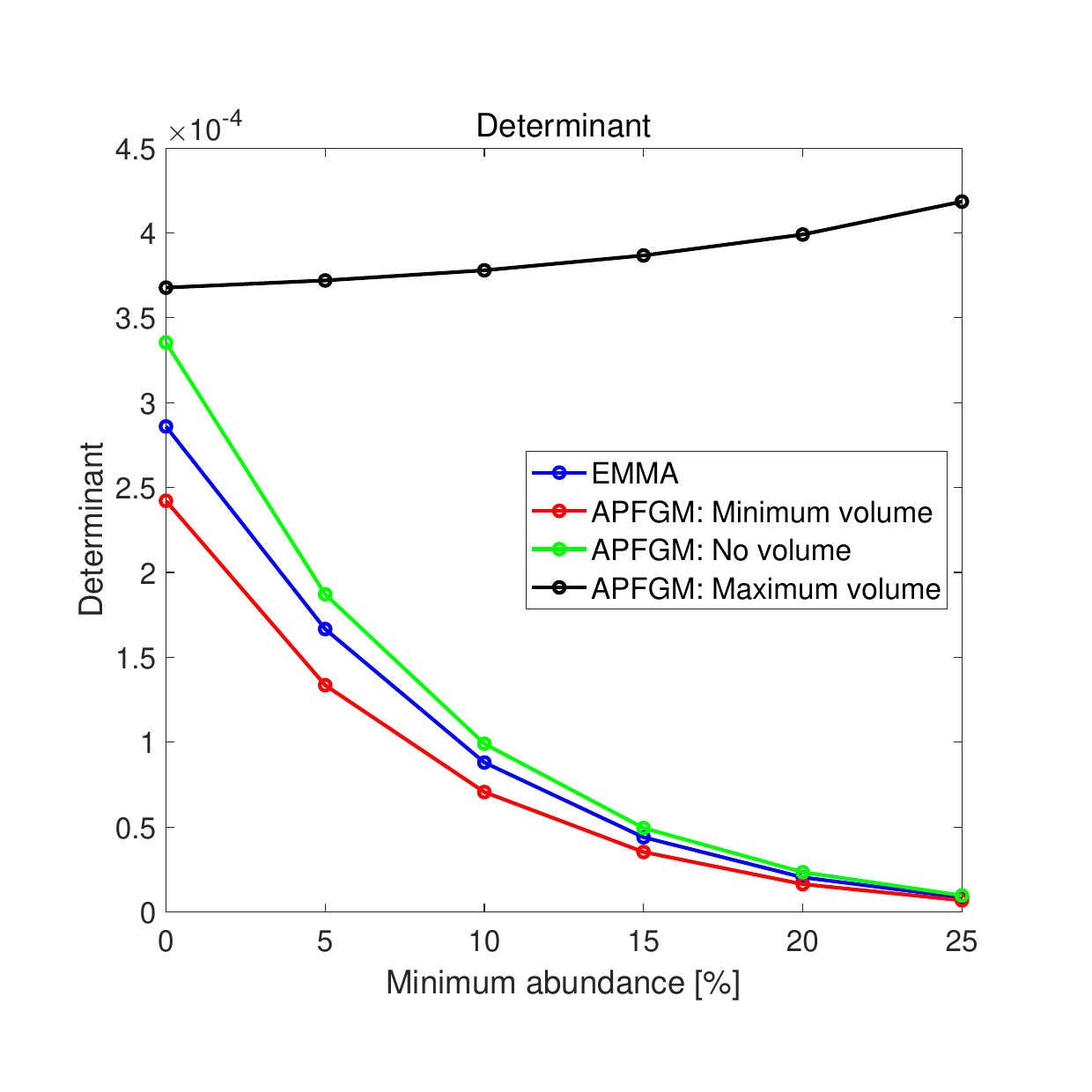}
 \caption{EMMA, APFGM with minimum volume, APFGM with no volume, and
APFGM with maximum volume: $\text{det}(\bm{G}\bm{G}^T)$.}\label{F: Determinant for Basis Matrix}
    \end{figure}

\begin{table}[H]
    \caption{The $\text{det}(\bm{G}\bm{G}^T)$ of EMMA, APFGM with minimum volume, APFGM with no volume, and
APFGM with maximum volume under different mixed level.}
    \label{T: Determinant for Basis Matrix}
    \centering
    \begin{tabular}{ccccc}
    \hline
    Mixed level & EMMA & APFGM min & APFGM no &  APFGM max
    \\
    \hline
     0.00    &   0.2860 &   0.2422  &  0.3355  &  0.3678\\
  0.05   & 0.1666  &  0.1336   & 0.1871  &  0.3721\\
  0.10   & 0.0880  &  0.0706   & 0.0991  &  0.3780\\
 0.15  &  0.0439 &  0.0353  &  0.0494   &0.3868\\
0.20  &  0.0203  &  0.0163  &  0.0234 &   0.3991\\
 0.25   &  0.0085 &   0.0068  &  0.0096  &  0.4185 
 \\
 \hline
    \end{tabular}
\end{table}
    
\section{Conclusion}\label{S: conclusion}

To conclude, in this paper, we propose maximum volume constrained EMA (MVC-EMA) for highly mixed GSD data. We prove that MVC-EMA is unique under the sufficient scattered conditions, a
new NMF identification criterion that is careful tweak of the existing volume minimization criterion in minimum volume constrained NMF. We introduce APFGM which makes use of alternative projected fast gradient methods. Experimental results show that APFGM with maximum volume can effective deal with highly mixed GSD data than APFGM with minimum volume, APFGM with no volume, and EMMA. However, as a new method, there are still some important issues that are in need of further research. First, MVC-EMA
needs to be tested on some real SGD to fully evaluate its performance. Second, algorithms which allows the balancing parameter to be any number need to be investigated.

\section*{Acknowledgments}

Zhongming Chen is partially supported by Natural Science Foundation of Zhejiang Province (No. LY22A010012) and Natural Science Foundation of Xinjiang Uygur Autonomous Region (No. 2024D01A09).

\section*{Competing Interests}

No potential competing interest was reported by the authors.

\bibliography{references.bib}

\end{document}